\documentclass[journal]{IEEEtran}
\IEEEoverridecommandlockouts

\usepackage{amsmath,amssymb,amsthm}
\usepackage{mathtools}
\usepackage{enumitem}
\setlist[description]{style=nextline,leftmargin=1em,labelindent=0pt,parsep=2pt,itemsep=4pt}
\usepackage[hidelinks]{hyperref}
\usepackage{url}
\usepackage{graphicx}
\usepackage{booktabs}
\usepackage{array}
\usepackage{multirow}
\usepackage{etoolbox}
\usepackage{xcolor}
\usepackage{listings}
\usepackage{microtype}
\usepackage{placeins}
\usepackage{balance}
\usepackage{xspace}
\usepackage{orcidlink}
\usepackage[ruled,linesnumbered]{algorithm2e}
\usepackage{tikz}
\usetikzlibrary{positioning,shapes.geometric,arrows.meta,fit,calc,decorations.pathreplacing,backgrounds}

\graphicspath{{./figures/}{./}}

\newcommand{\qsim}{\textsc{qsimadv}\xspace}
\newcommand{\cuq}{\textsc{cuQuantum}\xspace}
\newcommand{\custate}{\textsc{cuStateVec}\xspace}
\newcommand{\custn}{\textsc{cuTensorNet}\xspace}
\newcommand{\aer}{\textsc{Aer}\xspace}
\newcommand{\lightning}{\textsc{Lightning}\xspace}
\newcommand{\lightninggpu}{\textsc{Lightning-GPU}\xspace}
\newcommand{\lightningkokkos}{\textsc{Lightning-Kokkos}\xspace}
\newcommand{\juqcs}{\textsc{JUQCS}\xspace}
\newcommand{\qsimg}{qsim\xspace}
\newcommand{\nwqsim}{\textsc{NWQ-Sim}\xspace}
\newcommand{\openqasm}{OpenQASM\,3.0\xspace}

\newcommand{\ella}{\emph{Ella}\xspace}
\newcommand{\setonix}{\emph{Setonix}\xspace}

\newcommand{\bigO}[1]{\mathcal{O}\!\left(#1\right)}
\newcommand{\Qloc}{Q_{\text{loc}}}
\providecommand{\ket}[1]{\ensuremath{\left|#1\right\rangle}}

\newcommand{\code}[1]{\texttt{#1}}

\newcommand{\good}[1]{\textbf{#1}}

\makeatletter
\let\qsim@orig@makecaption\@makecaption
\renewcommand{\fnum@table}{Table~\thetable}
\long\def\@makecaption#1#2{%
  \ifdefstring{\@captype}{table}{%
    \vskip\abovecaptionskip
    {\normalfont\footnotesize\noindent #1. #2\par}%
    \vskip\belowcaptionskip
  }{%
    \qsim@orig@makecaption{#1}{#2}%
  }%
}
\makeatother

\newtheorem{theorem}{Theorem}
\newtheorem{lemma}{Lemma}
\newtheorem{corollary}{Corollary}

\hypersetup{
  pdftitle={QSimAdv: A Late-Bound, Vendor-Agnostic Architecture for High-Performance Quantum-Circuit Simulation},
  pdfauthor={Shusen Liu, Pascal Jahan Elahi, Wenyun Sun, Runyao Duan, Shenjin Lv, and Ugo Varetto},
  pdfsubject={High-performance exact quantum-circuit simulation},
  pdfkeywords={quantum circuit simulation, vendor-agnostic simulation, stabiliser formalism, state-vector simulation, multi-GPU, MPI, reproducibility}
}

\begin{document}

\title{QSimAdv: A Late-Bound, Vendor-Agnostic\\Architecture for High-Performance\\Quantum-Circuit Simulation}

\author{%
\IEEEauthorblockN{
Shusen Liu\orcidlink{0000-0002-0794-1132}\IEEEauthorrefmark{1}\IEEEauthorrefmark{2},
Pascal Jahan Elahi\orcidlink{0000-0002-6154-7224}\IEEEauthorrefmark{1}\IEEEauthorrefmark{2},
Wenyun Sun\orcidlink{0000-0002-2049-3960}\IEEEauthorrefmark{3},
Shenjin Lv\orcidlink{0009-0004-3744-6871}\IEEEauthorrefmark{4},
Xiaohan Shan\orcidlink{0000-0001-7468-4853} \IEEEauthorrefmark{5},
Ugo Varetto\orcidlink{0000-0002-7696-0345}\IEEEauthorrefmark{1}\IEEEauthorrefmark{2}
}

\IEEEauthorblockA{
\IEEEauthorrefmark{1}Pawsey Supercomputing Research Centre, Australia
}

\IEEEauthorblockA{
\IEEEauthorrefmark{2}The University of Western Australia, Australia
}
\IEEEauthorblockA{
\IEEEauthorrefmark{3}Nanjing University of Information Science \& Technology, China
}
\IEEEauthorblockA{
\IEEEauthorrefmark{4}Sina Weibo, China
}
\IEEEauthorblockA{
\IEEEauthorrefmark{5}Beijing Zhongke Liangshu Technology Co., Ltd, China
}

\IEEEauthorblockA{
Email: pascal.elahi@pawsey.org.au
}
}

\maketitle

\begin{abstract}
Portability in high-performance quantum-circuit simulation need not begin at the kernel. We present \qsim, which makes late binding, rather than a common kernel, the basis of vendor independence. Representation, operator lowering, and data placement are bound only when their required inputs become available. Before full-state allocation, circuit, noise, and output inspection can route eligible generic sampled-count requests to a stabiliser tableau; explicitly requested representations remain fixed. For full-state execution, backend constraints shape fusion; an ordered fused operator binds to a native lowering only after its physical targets are known. A first-class logical-to-physical layout map records non-canonical order across local and rank-address bits, so the dispatcher moves nonlocal targets only on demand. GPU, CPU, and Message Passing Interface (MPI) backends share these semantics while retaining native execution paths.

We realise this design on NVIDIA GH200 and AMD MI250X/EPYC systems across local and distributed execution. With matched complex 32-bit floating-point state storage, \qsim leads both \aer Hopper configurations at $N=32$ and \aer's HIP backend at four shared MI250X sizes from $N=24$ to 30. Strong scaling exposes platform dependence: on \setonix, \qsim leads both GPU and CPU comparisons at every measured rank, achieving $3.4\times$ and $2.8\times$ speedups, respectively, from one to eight; neither GH200 path speeds up at eight ranks. Weak scaling reaches 256 ranks with 2~TiB GPU and 1~TiB CPU states. Together, these results support the central claim: portability can reside above the kernel boundary while execution remains native and extends across distributed memory.
\end{abstract}


\section{Introduction}\label{sec:intro}

Quantum-circuit simulators provide full-state and intermediate-amplitude references that quantum hardware cannot expose directly. Exact simulation is also a demanding high-performance-computing workload: an $N$-qubit state vector contains $2^N$ amplitudes, each generic gate touches the full state, and simulations beyond one device must move parts of that state across a communication fabric.

Several widely used GPU simulators obtain their fastest NVIDIA paths through the closed \cuq SDK, particularly its \custate state-vector library~\cite{bayraktar2023cuquantum}. These include Qiskit \aer's GPU backend~\cite{aleksandrowicz2019qiskit,wood2024aerprivacy}, PennyLane \lightninggpu~\cite{asadi2024lightning}, Google \qsimg-GPU~\cite{isakov2021qsim}, and CUDA-Q~\cite{brown2026multigpu}. That route is effective on NVIDIA hardware, but it has no equivalent quantum-simulation API on the AMD system evaluated here. \setonix~\cite{setonix} supplies HIP, matrix and tensor libraries, and GPU-aware MPI, yet its ROCm stack does not provide the gate, index-swap, and measurement operations exposed by \custate~\cite{rocm721release}. Translating CUDA syntax to HIP therefore leaves the simulator algorithms themselves to be implemented.

\qsim treats representation choice and hardware lowering as one late-binding problem. It owns three decisions above the vendor API: semantic routing, fused-operator execution, and the placement of logical qubits in local memory or across MPI ranks. Vendor backends implement four primitives (matrix multiplication, index permutation, sampling, and distributed exchange) without choosing the representation or defining the logical layout. This separation preserves the logical semantics across vendors while leaving each backend free to use native libraries and kernels.

The generic full-state lowering has a simple interpretation. View the $2^N$ amplitudes as an $N$-axis tensor, with one axis per qubit. A fused $k$-qubit gate acts on $k$ selected axes. Bringing those axes together turns the update into a matrix multiplication; if the axes are not moved back afterwards, a small layout map records their new positions. Specialised kernels preserve the same target and layout semantics without necessarily performing this transform. In a distributed state vector, the high index bits select the MPI rank. Moving a gate target from a rank bit to a local bit is then analogous to bringing a requested item into a cache, which connects the placement problem to offline paging.

We employ the design as \emph{late-bound} because each architectural choice is committed at the last stage with the information needed to preserve the requested semantics. Circuit inspection resolves the representation before a large state is allocated. For a request that remains on the full-state path, the fallback capability selects a backend family and its feasibility limits, including the fusion cap, before fusion. Once a fused operator and its physical targets are known, that backend selects the concrete kernel or transport. Local and distributed layout maps retain non-canonical axis orders until an operation requires materialisation.

The source tree has been evaluated on HPC \ella, a Grace Hopper GH200 system with CUDA and HPC-X OpenMPI, and HPC \setonix, an MI250X/EPYC system with ROCm and Cray MPICH. The experiments cover semantic routing, single-device CPU and GPU execution, strong scaling, and weak scaling to 256 ranks.

\subsection{Contributions}\label{sec:intro-thesis}

This paper makes four contributions.

\begin{enumerate}
\item \textbf{Semantic routing (R1).} For generic \code{dense} and \code{gpu} requests that ask for sampled counts, \qsim inspects the circuit before allocating a state vector. Eligible Clifford circuits with Pauli noise use a stabiliser tableau; explicit density-matrix, matrix-product-state (MPS), and tensor-network requests keep their chosen representation. At $N=30$, routing the noisy Greenberger--Horne--Zeilinger (GHZ) benchmark is about $78{,}000\times$ faster than general execution, and the tableau path reaches $N=1000$ in under a second.

\item \textbf{A fused-operator interface (R2).} The CPU state-vector path, the CUDA and HIP GPU state-vector paths, rank-local state-vector execution within MPI, and the density-matrix wrapper all accept the same ordered targets, operator, and layout map. The generic implementation combines an axis permutation with general matrix--matrix multiplication (GEMM), while diagonal, controlled, and small-width gates use specialised kernels. On the Hopper random-circuit sweep, \qsim overtakes \aer without \custate from $N=28$, with an \aer-to-\qsim time ratio of 2.2 at $N=32$; on the four MI250X random-brickwork cells shared by all GPU tools, at $N\in\{24,26,28,30\}$, the ratio ranges from 1.15 to 2.16.

\item \textbf{Distributed layout as explicit state (R3).} A logical-to-physical map makes deferred cross-rank placement explicit, and a proved invariant keeps it correct. For a fixed gate stream, the ideal all-frame, full-future demand model with uniform promotion costs reduces to offline paging, where B\'el\'ady's rule minimises the number of promotions. The production planner uses bounded lookahead, and the distributed interface reaches 256 ranks and a 2~TiB state.

\item \textbf{Backend lowering through a shared primitive contract (R4).} CUDA, HIP, Arm, x86, and MPI implementations preserve the same representation and layout contracts while choosing native matrix-multiplication, permutation, sampling, and exchange operations. The evaluation characterises these lowerings with a kernel decomposition on NVIDIA and a bounding host-region profile on AMD.
\end{enumerate}

\S\ref{sec:background}--\ref{sec:methods} give the background, architecture, and implementation, \S\ref{sec:eval} the evaluation, and \S\ref{sec:related}--\ref{sec:conclusion} the related work, implications, and conclusions; the appendices hold the artefact, measurement protocol, proofs, and complete tables.

\section{Background: representations, APIs, and backend primitives}\label{sec:background}

Three decisions shape \qsim's execution path: semantic routing selects a representation, the fused-operator interface defines full-state execution, and a layout map records logical-to-physical placement. Vendor backends realise these decisions through four primitives. The underlying tableau and matrix-multiplication algorithms are established methods; the contribution lies in how they are selected and joined across backends.

\paragraph{State and layout notation} A circuit $C$ contains gates, measurements, and noise operations on $N$ logical qubits; $\mathcal{N}$ denotes its per-operation noise channels, and $S$ is the number of sampled shots. A pure state is a unit vector $\ket{\psi}\in\mathbb{C}^{2^N}$. The same data can be viewed as a rank-$N$ tensor $\Psi$, with one binary axis for each qubit. Physical axis positions are numbered by the significance of their flat-index bits, starting with position 0 as the least significant. A layout map $L$ records the physical position of each logical qubit, so a permuted tensor still represents the same logical state.

For distributed execution, let $P=2^r$ be the number of MPI ranks, where $0\le r<N$. Each rank stores $\Qloc=N-r$ local index bits; the remaining $r$ high bits select the rank. The map $\pi_t$ extends $L$ across both local and rank-address bits after execution step $t$.

\paragraph{Fused gates and requests} Adjacent operations can be multiplied in circuit order to form a fused $k$-qubit gate $G\in\mathbb{C}^{2^k\times2^k}$. Its ordered target tuple is $V=(q_0,\ldots,q_{k-1})$, with $k=|V|$; $K_{\max}$ denotes the configured fusion cap. The matrix convention is little-endian within this tuple: $q_j$ supplies bit $j$ of a row or column index, so $q_0$ is the least-significant target. The generic \code{dense} and \code{gpu} entry points leave the representation open and name the fallback full-state path. Explicit density-matrix, MPS, and tensor-network requests commit their representation. If semantic routing selects the CPU-resident tableau, that representation-specific implementation takes precedence over the generic full-state fallback.

\subsection{Representation classes}\label{sec:bg-regimes}

\paragraph{State vectors} A state vector represents arbitrary pure-state circuits exactly up to floating-point error. It requires $\Theta(2^N)$ memory. Applying a generic dense fused $k$-qubit operator costs $\Theta(2^{N+k})$ arithmetic on a $2^k\times2^{N-k}$ view of the state, while diagonal and other structured operators can be applied in $\Theta(2^N)$ work. The exponential state size makes this representation the main target for hardware acceleration and distributed memory.

\paragraph{Tensor networks and MPS} Tensor-network methods avoid materialising the full state vector, but their cost depends on the sizes of intermediate tensors induced by the contraction order. In an MPS, this dependence appears through the bond dimensions needed to represent entanglement across cuts of the qubit chain~\cite{vidal2003efficient}. Such methods can be effective when those intermediate dimensions remain small. An MPS result becomes approximate when a bond cap or singular-value cutoff discards weight. \qsim includes an MPS engine and an NVIDIA-only \custn path; neither is evaluated below.

\paragraph{Stabiliser tableaux} Clifford gates map Pauli operators to Pauli operators. The Gottesman--Knill theorem therefore permits polynomial-time simulation of Clifford circuits with stabiliser preparations and Pauli measurements~\cite{gottesman1998heisenberg}; sampled Pauli noise preserves the same structure. \qsim uses the Aaronson--Gottesman tableau~\cite{aaronson2004improved}, which represents $N$ stabiliser and $N$ destabiliser generators by binary coordinates and phase bits. A non-Clifford operation leaves this representation class and requires another engine. Choosing a tableau instead of a state vector changes the asymptotic cost, not merely a backend constant.

\subsection{How simulators select a representation}\label{sec:bg-landscape}

Mature implementations exist for all three classes. State-vector engines include \aer~\cite{aleksandrowicz2019qiskit}, \lightning~\cite{asadi2024lightning}, \qsimg~\cite{isakov2021qsim}, Qulacs~\cite{suzuki2021qulacs}, QuEST~\cite{jones2019quest}, \juqcs~\cite{juqcs2018}, \nwqsim~\cite{li2023nwqsim}, and Atlas~\cite{xu2024atlas}. Stim provides stabiliser sampling~\cite{gidney2021stim}; Quimb~\cite{quimb2018}, TensorCircuit~\cite{tensorcircuit2023}, and \custn-backed engines~\cite{bayraktar2023cuquantum} provide tensor-network paths. Their selection policies differ.

Some tools require the user to choose an engine before the circuit is inspected. Qulacs and QuEST expose state-vector and density-matrix engines, \qsimg exposes state-vector, hybrid, and MPS engines~\cite{qsimv022}, and Stim is a tableau simulator. The user is responsible for matching the engine to the circuit.

Multi-method frameworks may choose automatically. \aer's default \code{method='automatic'} can select its stabiliser engine for an all-Clifford circuit~\cite{wood2024aerprivacy}; an explicit \code{method='statevector', device='GPU'} request keeps the requested method. \qsim inspects eligible sampled-count requests made through its generic \code{gpu} and \code{dense} entry points. These entry points describe a fallback full-state path rather than pinning a representation; explicit density-matrix, MPS, and tensor-network requests retain their requested semantics. Tab.~\ref{tab:feature-matrix} compares this scope with the other selected systems.

Routing before allocation matters most at the Clifford boundary. The parser can choose the polynomial tableau before creating an exponential state vector, and the selected representation then determines how noise, shots, and output are handled downstream.

\subsection{Vendor quantum APIs and toolkit primitives}\label{sec:bg-cuq}

Accelerator platforms expose useful operations at two levels. A quantum-simulation library such as \custate applies gates, swaps index bits, and samples a state vector; \custn contracts tensor networks~\cite{bayraktar2023cuquantum}. A lower-level toolkit instead provides matrix multiplication, tensor permutation, random numbers, reductions, and communication. NVIDIA supplies both levels. CUDA offers cuBLAS and cuTensor, while \cuq adds the quantum operations used by several simulators. AMD supplies the toolkit level through HIP, hipBLAS/hipBLASLt, hipTensor, rocRAND, and GPU-aware MPI, but the evaluated stack has no direct analogue of \custate or \custn.

\qsim targets the toolkit level. Its common interface describes gates, tensor axes, and layouts; CUDA, HIP, CPU, and MPI backends implement the required operations with their native libraries or kernels. Throughout, \emph{vendor-agnostic} means that the layers above the backend name logical qubits, tensor axes, matrices, and layout maps rather than vendor handles or operation descriptors; a backend still chooses its own libraries, arithmetic modes, kernels, and feasibility limits, provided that it preserves the selected representation and layout semantics. \S\ref{sec:design} defines this interface, and \S\ref{sec:eval} measures its behaviour on both vendor stacks.

\section{Architecture}\label{sec:design}

A run begins with a parsed circuit $C$, noise model $\mathcal{N}$, requested entry point, and shot count $S$. A generic full-state entry point also identifies its fallback CPU, GPU, or MPI capability, but not a concrete per-operator lowering. R1--R3 define the three above-backend decisions: representation, fused-operator semantics, and qubit layout. R4 defines the primitive contract and feasibility limits through which a selected full-state backend implements them.

\begin{figure*}[t]
\centering
\includegraphics[width=\textwidth]{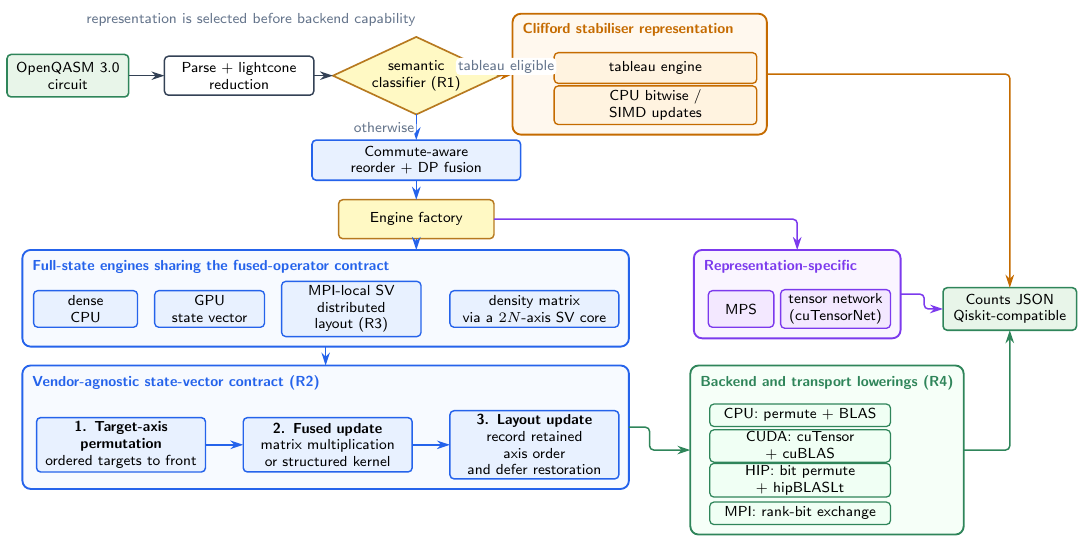}
\caption{System pipeline. Routing selects the representation before large state allocation. MPS and tensor-network engines retain representation-specific paths. Here SV denotes state vector and DP denotes dynamic programming.}
\label{fig:architecture}
\end{figure*}

\begin{description}
\item[(R1) Semantic routing.] A request through the generic \code{gpu} or \code{dense} entry point asks for sampled counts without pinning the representation. Before allocating a state vector, the router checks whether its circuit, noise, and output semantics support the stabiliser path. An explicit benchmarking switch can disable this route. Density-matrix, MPS, and tensor-network requests retain their requested method and numerical semantics (\S\ref{sec:bg-regimes}).

\item[(R2) Fused-operator semantics and local layout.] Full-state engines expose
\[
\mathrm{apply}(L,\Psi,G,V)\rightarrow(L',\Psi'),
\]
where $V$ is an ordered tuple of logical targets, $G$ is the corresponding fused operator, $L$ bijectively maps logical qubits to physical index-bit positions, and $\Psi$ is the physically stored state. A backend may use a specialised lowering or a generic transpose--transpose--GEMM--transpose (TTGT) path~\cite{springer2018tccg}. If a lowering leaves the state in a non-canonical order, $L'$ records that order.

\item[(R3) Distributed placement invariant.] MPI engines extend the layout map across local and rank-address bits, so that after step $t$ the map $\pi_t$ of \S\ref{sec:background} sends each logical qubit to the physical index bit that stores it. The map is a bijection: every qubit occupies one position and every position holds one qubit. If a gate target lies in a rank-address bit, the dispatcher exchanges it with a local bit and records the new placement in $\pi_t$. Later gates resolve their targets through the map, and output routines either decode the final placement or request canonical order.

\item[(R4) Backend lowering through a shared primitive contract.] CUDA, HIP, CPU, and MPI implementations provide four primitives (matrix multiplication, index permutation, sampling, and distributed exchange) plus specialised gate kernels. If routing retains a full-state path, the fallback capability selects the backend family before fusion. That family supplies feasibility limits, including the fusion cap and distributed chunk width. Once a fused operator and its physical targets are known, the backend selects the concrete kernel or transport.
\end{description}

Fig.~\ref{fig:architecture} traces the pipeline from an \openqasm circuit~\cite{cross2022openqasm} to sampled counts. Parsing first removes operations that cannot influence any measured qubit, that is, operations outside the measured outputs' causal cone. Eligible Clifford circuits with Pauli noise then take the tableau path without allocating a $2^N$ state vector. For a circuit that remains on the full-state path, the engine factory selects a compatible backend family and its limits before commute-aware reordering and fusion. The resulting blocks execute through the R2 fused-operator interface; MPS and tensor-network engines keep their representation-specific semantics.

\section{Implementation}\label{sec:methods}

The implementation follows the architecture's late-binding order: it resolves the representation before allocation, selects the full-state backend family and its limits before fusion, and chooses each concrete per-operator lowering after the fused block and its physical targets are known.

\subsection{Semantic routing and engine selection (R1)}\label{sec:methods-stab}

The generic \code{gpu} and \code{dense} entry points specify a fallback full-state path but leave the representation open. For sampled-count requests, R1 may instead select a tableau when the circuit is eligible. Explicit representation requests keep their selected method. If routing retains full-state execution, the fallback capability selects a compatible CPU, GPU, or MPI backend family and the router checks its size limits.

A circuit is tableau-eligible when every operation remains within Clifford dynamics: Clifford gates, rotations equivalent to Clifford gates, Pauli measurements, and supported Pauli noise channels. A Pauli channel is a classical mixture of tensor-product Pauli errors. For each shot, \qsim samples one error from every such mixture. Pauli operators are themselves Clifford operators, so inserting the sampled errors into a Clifford circuit leaves the entire trajectory Clifford and exactly representable by a tableau~\cite{gottesman1998heisenberg}. The channel output is generally a mixture rather than one stabiliser state, but the ensemble of tableau trajectories is exact: averaging their conditional measurement distributions gives the noisy output distribution. Pauli noise therefore retains the polynomial stabiliser path rather than forcing a full-state representation.

The noise need not have one global error rate. Channel parameters are attached to individual circuit operations, and the one-qubit Pauli form permits unequal $X$, $Y$, and $Z$ probabilities. The implemented checker recognises bit flip, phase flip, bit-phase flip, depolarising noise, phase damping in its $I/Z$-mixture form, and biased one-qubit Pauli noise; multi-qubit depolarising noise is uniform over its non-identity Pauli strings. Amplitude damping and generic custom Kraus channels are not eligible, and the router does not inspect a custom Kraus list to infer that it happens to define a Pauli mixture. A single pass over the parsed circuit performs the eligibility check; any unsupported operation keeps the request on the full-state path. The tableau is CPU-resident, and its arrays of 64-bit integers are updated through AVX-512, AVX2, SVE/SVE2, NEON, or scalar kernels. A CUDA or HIP request that takes this route therefore allocates no GPU state vector.

For an eligible circuit, \qsim uses an Aaronson--Gottesman tableau~\cite{aaronson2004improved}. The tableau has $N$ destabiliser rows and $N$ stabiliser rows, each containing $2N$ binary symplectic coordinates and one phase bit. The implementation stores each symplectic column and the phase column in its own array of 64-bit unsigned integers. Within any such array, row $s$, for $0\le s<2N$, occupies bit $s\bmod 64$ of element $\lfloor s/64\rfloor$; each column therefore uses $\lceil 2N/64\rceil$ array elements. A one- or two-qubit Clifford gate applies a fixed sequence of bitwise operations to a constant number of columns.

\begin{samepage}
Let $W_{\mathrm{Clifford}}(N)$ count the 64-bit integer updates needed for one such gate. Then
\begin{equation}
W_{\mathrm{Clifford}}(N)
  = \Theta\left(\left\lceil\frac{2N}{64}\right\rceil\right)
  = \Theta(N).
\label{eq:cost-cliff-gate}
\end{equation}
\end{samepage}
For measurement, the implementation forms a row-oriented copy in which a generator's $N$ $X$-coordinates and $N$ $Z$-coordinates occupy $2\lceil N/64\rceil$ 64-bit integers. A measurement can combine $\bigO{N}$ such generator rows. Charging the $|C|$ circuit operations at the gate-update cost and adding measurement and initialisation work gives the conservative per-shot work, measured in the same unit:
\begin{equation}
W_{\mathrm{shot}}(C,N,n_{\mathrm{meas}})
  = \bigO{|C|\cdot N+(n_{\mathrm{meas}}+1)\cdot N^2},
\label{eq:cost-stab-shot}
\end{equation}
where $n_{\mathrm{meas}}$ counts measurement operations; for the terminal-readout circuits evaluated it equals the number of measured qubits. The added unit accounts for constructing the tableau. The fast path below constructs the noiseless tableau once and reuses it across shots, amortising this initialisation cost.

\begin{algorithm}
\SetAlgoLined
\KwIn{Circuit $C$, classified as tableau-eligible; Clifford gates $g$ and per-operation Pauli channels $\mathcal{E}$. Each sample $\sigma\sim\mathcal{E}$ is a Pauli operator drawn from the corresponding channel mixture.}
\KwOut{Per-shot measurement record $\{b_q\}$ for measured qubits $q$.}
$\mathcal{T} \leftarrow$ initialise the tableau for $\ket{0^N}$\;
\For{op in $C$}{
  \uIf{op is a Clifford gate $g$}{apply\_$g$($\mathcal{T}$)\;
    \If{$g$ carries a Pauli channel $\mathcal{E}$}{sample
    $\sigma\sim\mathcal{E}$; apply\_$\sigma$($\mathcal{T}$)\;}}
  \uElseIf{op is a standalone Pauli channel $\mathcal{E}$}{sample
  $\sigma\sim\mathcal{E}$; apply\_$\sigma$($\mathcal{T}$)\;}
  \uElseIf{op is a measurement on qubit $q$}{$b_q \leftarrow$
  measure($\mathcal{T}$, $q$)\;}
}
\Return{$\{b_q\}$}\;
\caption{One stabiliser shot under semantic routing. $\mathcal{T}$ denotes the Aaronson--Gottesman tableau; $\mathcal{E}$ denotes a per-operation Pauli channel.}
\label{alg:tableau}
\end{algorithm}

For terminal computational-basis measurements, sampled Pauli errors need not trigger a full tableau replay. Clifford gates map Pauli operators to Pauli operators, so the errors can be propagated to the end of the circuit and combined into one Pauli frame,
\[
F = \omega \prod_{j=0}^{N-1} X_j^{f_{X,j}} Z_j^{f_{Z,j}}, \qquad \mathbf{f}_X,\mathbf{f}_Z \in \{0,1\}^{N},
\]
where $f_{X,j}$ and $f_{Z,j}$ record the $X$ and $Z$ components on qubit $j$ and $\omega$ is an irrelevant global phase. (The $X_jZ_j$ ordering fixes the sign of a $Y$ component.)

This reduction assumes one terminal joint computational-basis readout. All sampled noise must precede the first recorded measurement, with no reset, adaptive classical control, or mid-circuit dependence on a measurement outcome.

Let $\mathcal{M}=(q_0,\ldots,q_{m-1})$ be the ordered list of terminally measured qubits, and let $\mathbf{f}_{X,\mathcal{M}}=(f_{X,q_0},\ldots,f_{X,q_{m-1}})$ be the restriction of the frame's $X$ component to those qubits. For every measured output string $z\in\{0,1\}^{m}$,
\begin{equation}
\Pr\bigl(z \mid F\bigr)
=
\Pr_{\mathrm{noiseless}}\bigl(z\oplus\mathbf{f}_{X,\mathcal{M}}\bigr),
\label{eq:pauli-frame-identity}
\end{equation}
where $\oplus$ is bitwise XOR. The $Z$ component changes only phases and therefore leaves computational-basis probabilities unchanged. Each $X$ component on a measured qubit flips the corresponding output bit; on unmeasured qubits it only permutes traced-out outcomes.

Accordingly, XORing a noiseless output sample with $\mathbf{f}_{X,\mathcal{M}}$ yields a sample from the trajectory conditioned on $F$. Averaging Eq.~\eqref{eq:pauli-frame-identity} over frames drawn from the Pauli-channel mixture recovers the noisy distribution $\Pr_{\mathrm{noisy}}(z)=\mathbb{E}_F[\Pr(z\mid F)]$.

\qsim constructs the noiseless tableau once and propagates up to 64 independent Pauli frames together. Bit $b$ of each 64-bit integer records the $X$ or $Z$ component for shot $b$, so one Clifford update advances the whole batch. The implementation then samples one ideal terminal record and applies the XOR correction for each active frame. Thus $S$ shots require $\lceil S/64\rceil$ batched frame walks and $S$ noiseless samples, with no nested trajectory loop. Tableau-eligible circuits outside this terminal-readout regime use Algorithm~\ref{alg:tableau}. Appendix~\ref{appendix:qec-detail} describes the synthetic benchmark, which uses the fast path.

\subsection{Fused-operator interface and lazy TTGT (R2)}\label{sec:methods-axisshift}

For full-state execution, R2 fixes the semantics of a fused operator while allowing several physical kernels. Given $(L,\Psi,G,V)$, a backend resolves the ordered physical target positions $A=(L(q_0),\ldots,L(q_{k-1}))$, applies $G$ in the little-endian order fixed by $V$, and returns a state-layout pair that preserves the logical state. Diagonal, controlled, small-width in-place, and custom dense kernels handle common cases. The generic TTGT-style fallback moves the target bits into the matrix-row block, reshapes the state as $M\in\mathbb{C}^{2^k\times2^{N-k}}$, and computes $M'=GM$. The updated layout map records the retained forward permutation, allowing the inverse data permutation to remain deferred (Fig.~\ref{fig:axis-shift-tensor}).

Fusion and placement operate on fused blocks and layout metadata; they do not assume that every block executes through GEMM.

A tensor in permuted physical order remains a valid state representation as long as $L$ records the physical index-bit position of each logical qubit. For the generic permutation--GEMM lowering, applying $G$ can be written as
\[
\Psi \xmapsto{\;\Pi_A\;} \widetilde{\Psi} \xmapsto{\;G \otimes I_{2^{N-k}}\;} \widetilde{\Psi}' \xmapsto{\;\Pi_A^{-1}\;} \Psi'_{\mathrm{rest}},
\]
where $\Pi_A$ sends the target at $L(q_j)$ to physical position $N-k+j$. Thus $q_j$ becomes bit $j$ of the matrix-row index; in a row-major tensor display the leading target axes appear as $q_{k-1},\ldots,q_0$. The state $\Psi'_{\mathrm{rest}}$ restores the input order. The restored branch returns $(L,\Psi'_{\mathrm{rest}})$. The lazy branch stops at $\widetilde{\Psi}'$ and returns $(L',\widetilde{\Psi}')$, with $L'$ recording the retained permutation.

Under this reshape, if $b_j$ is the value of target $q_j$, its row number is $\sum_{j=0}^{k-1}b_j2^j$. The columns enumerate assignments to the remaining bits. This is the basis order used by $G$, so left multiplication by $G$ applies the fused gate independently to every column.

The data permutation $\Pi_A$ rearranges amplitudes; let $\alpha_A$ denote the same rearrangement at the level of physical bit positions. Thus, a logical qubit stored at position $b$ before the permutation is stored at position $\alpha_A(b)$ afterwards. Algorithm~\ref{alg:axis-shift} records this change as $L'=\alpha_A\circ L$. This is the gate-application form of TTGT~\cite{springer2018tccg}, with the trailing transpose represented in $L'$ instead of executed. We call the generic fallback \emph{lazy TTGT}; specialised R2 kernels need not use GEMM.

\begin{figure}[t]
\centering
\definecolor{axisCoreFill}{HTML}{E8F1FF}
\definecolor{axisCoreStroke}{HTML}{2563EB}
\definecolor{axisTargetFill}{HTML}{FFF3E0}
\definecolor{axisTargetStroke}{HTML}{C56B00}
\definecolor{axisDataFill}{HTML}{E8F5E9}
\definecolor{axisDataStroke}{HTML}{2F855A}
\definecolor{axisProcFill}{HTML}{FFF9C4}
\definecolor{axisProcStroke}{HTML}{B7791F}
\definecolor{axisInk}{HTML}{334155}
\definecolor{axisMuted}{HTML}{52606D}

\begin{tikzpicture}[
  x=1cm,
  y=1cm,
  font=\sffamily\scriptsize,
  every node/.style={align=center, outer sep=0pt},
  axisNode/.style={
    rectangle, rounded corners=1.5pt,
    draw=axisCoreStroke, line width=0.65pt, fill=axisCoreFill,
    minimum width=0.72cm, minimum height=0.48cm,
    inner xsep=2pt, inner ysep=1.5pt
  },
  targetAxis/.style={
    axisNode, draw=axisTargetStroke, line width=0.9pt, fill=axisTargetFill,
    text=axisTargetStroke
  },
  dataBox/.style={
    rectangle, rounded corners=2pt,
    draw=axisDataStroke, line width=0.7pt, fill=axisDataFill,
    minimum height=0.72cm,
    inner xsep=3pt, inner ysep=2pt
  },
  gateBox/.style={
    rectangle, rounded corners=2pt,
    draw=axisTargetStroke, line width=0.7pt, fill=axisTargetFill,
    minimum height=0.78cm, text width=1.20cm,
    inner xsep=3pt, inner ysep=2pt
  },
  matrixBox/.style={
    rectangle, rounded corners=2pt,
    draw=axisCoreStroke, line width=0.7pt, fill=axisCoreFill,
    minimum height=0.78cm, text width=1.55cm,
    inner xsep=3pt, inner ysep=2pt
  },
  layoutBox/.style={
    rectangle, rounded corners=2pt,
    draw=axisProcStroke, line width=0.7pt, fill=axisProcFill,
    minimum height=0.82cm, text width=2.55cm,
    inner xsep=3pt, inner ysep=2pt
  },
  deferBox/.style={
    rectangle, rounded corners=2pt,
    draw=axisTargetStroke, line width=0.65pt, fill=axisTargetFill!70,
    minimum height=0.45cm, text width=4.95cm,
    inner xsep=3pt, inner ysep=1.5pt
  },
  stageGroup/.style={
    rectangle, rounded corners=3pt,
    draw=axisInk!28, line width=0.55pt, fill=axisCoreFill!18,
    inner sep=0.14cm
  },
  stageTitle/.style={
    anchor=west, font=\sffamily\scriptsize\bfseries,
    text=axisInk, inner sep=0pt
  },
  secondary/.style={
    font=\sffamily\scriptsize, text=axisMuted,
    inner sep=0pt
  },
  flow/.style={
    -{Stealth[length=2.0mm,width=1.35mm]},
    draw=axisInk!82, line width=0.72pt
  },
  flowLabel/.style={
    font=\sffamily\scriptsize, text=axisInk!85,
    fill=white, inner sep=1.5pt
  }
]


\node[stageTitle] at (0.10,0.00) {1. Stored rank-$N$ tensor};
\node[secondary, anchor=east, text=axisTargetStroke] at (8.10,0.00)
  {identity layout $L(q)=q$; targets $V=(2,4)$};
\node[dataBox, text width=1.45cm] (inputState) at (1.05,-0.75)
  {state tensor\\$\Psi$};
\node[axisNode, minimum width=1.55cm] (a0) at (2.65,-0.75) {$N-1,\ldots,5$};
\node[targetAxis] (a1) at (3.85,-0.75) {$4$};
\node[axisNode]   (a2) at (4.60,-0.75) {$3$};
\node[targetAxis] (a3) at (5.35,-0.75) {$2$};
\node[axisNode]   (a4) at (6.10,-0.75) {$1$};
\node[axisNode]   (arest) at (6.85,-0.75) {$0$};
\coordinate (s1left) at (0.10,-0.75);
\coordinate (s1right) at (8.10,-0.75);

\node[stageTitle] at (0.10,-2.20) {2. Permute and flatten};
\node[targetAxis] (r0) at (0.55,-2.90) {$4$};
\node[targetAxis] (r1) at (1.30,-2.90) {$2$};
\node[axisNode, minimum width=1.55cm] (c0) at (2.60,-2.90) {$N-1,\ldots,5$};
\node[axisNode]   (c1) at (3.80,-2.90) {$3$};
\node[axisNode]   (c2) at (4.55,-2.90) {$1$};
\node[axisNode]   (crest) at (5.30,-2.90) {$0$};
\draw[draw=axisInk!55, line width=0.65pt]
  (1.73,-2.59) -- (1.73,-3.19);
\node[secondary] (rowLabel) at (0.93,-3.32) {$2^k$ rows};
\node[secondary] (columnLabel) at (3.85,-3.32) {$2^{N-k}$ columns};
\node[matrixBox] (mview) at (7.10,-2.90)
  {$M$\\$2^k\times 2^{N-k}$};
\draw[flow] (crest.east) -- (mview.west);
\coordinate (s2left) at (0.10,-2.90);
\coordinate (s2right) at (8.10,-2.90);

\node[stageTitle] at (0.10,-4.35) {3. Apply $G$ to the row index};
\node[gateBox]   (gate) at (1.65,-5.08)
  {$G$\\$2^k\times 2^k$};
\node[secondary, font=\sffamily\normalsize] at (2.72,-5.08) {$\times$};
\node[matrixBox] (matIn) at (3.82,-5.08)
  {$M$\\$2^k\times 2^{N-k}$};
\node[secondary, font=\sffamily\normalsize] at (4.93,-5.08) {$=$};
\node[matrixBox] (matOut) at (6.05,-5.08)
  {$M'$\\$2^k\times 2^{N-k}$};
\coordinate (s3left) at (0.10,-5.08);
\coordinate (s3right) at (8.10,-5.08);

\node[stageTitle] at (0.10,-6.50) {4. Updated state and layout};
\node[layoutBox] (layout) at (2.00,-7.28)
  {layout metadata\\$L'=\alpha_A\circ L$};
\node[secondary, font=\sffamily\scriptsize] at (3.68,-7.28) {and};
\node[dataBox, text width=3.65cm] (outputState) at (6.00,-7.28)
  {stored tensor $\widetilde{\Psi}'$\\row-major axis order $(4,2,N-1,\ldots,5,3,1,0)$};
\node[deferBox] (defer) at (4.10,-8.15)
  {$\Pi_A^{-1}$ is deferred; no immediate data movement};
\coordinate (s4left) at (0.10,-7.55);
\coordinate (s4right) at (8.10,-7.55);

\begin{scope}[on background layer]
  \node[stageGroup, fit=(s1left)(s1right)(inputState)(a0)(a1)(a2)(a3)(a4)(arest)]
    (stage1) {};
  \node[stageGroup,
    fit=(s2left)(s2right)(r0)(r1)(c0)(c1)(c2)(crest)(rowLabel)(columnLabel)(mview)]
    (stage2) {};
  \node[stageGroup, fit=(s3left)(s3right)(gate)(matIn)(matOut)]
    (stage3) {};
  \node[stageGroup, fit=(s4left)(s4right)(layout)(outputState)(defer)]
    (stage4) {};
\end{scope}

\draw[flow] (stage1.south) -- node[flowLabel, right=2pt]
  {$\Pi_A$: form the target row block} (stage2.north);
\draw[flow] (stage2.south) -- node[flowLabel, right=2pt]
  {apply $M'=GM$} (stage3.north);
\draw[flow] (stage3.south) -- node[flowLabel, right=2pt]
  {reinterpret the same buffer} (stage4.north);

\end{tikzpicture}
\caption{Lazy TTGT for the identity layout $L(q)=q$ and targets $V=(2,4)$. Axes are shown from most to least significant, so the target row block is $(4,2)$ although $q_0=2$ is its least-significant bit; $L'$ records the retained permutation.}
\label{fig:axis-shift-tensor}
\end{figure}
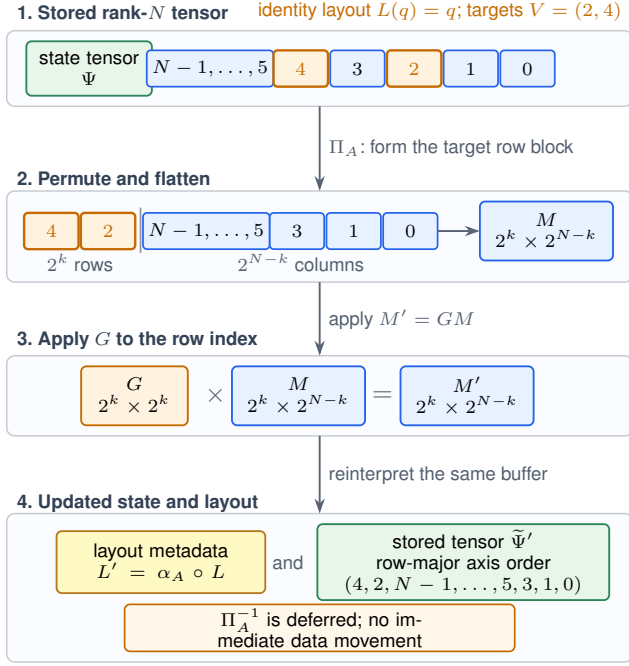

\begin{algorithm}
\SetAlgoLined
\KwIn{State tensor view $\Psi$ with layout map $L$; gate matrix $G$ on ordered logical target list $V=(q_0,\ldots,q_{k-1})$, with $q_0$ least significant; $\mathit{mode}\in\{\mathrm{dense},\mathrm{diag}\}$}
\KwOut{Updated layout map $L'$ and stored state view $\widetilde{\Psi}'$}
$A \leftarrow (L(q_0),\ldots,L(q_{k-1}))$\;
\eIf{$\mathit{mode}=\mathrm{diag}$}{
  $\widetilde{\Psi}' \leftarrow$ multiply each amplitude by the diagonal entry indexed by the target bits at positions $A$\;
  $L' \leftarrow L$\;
}{
  choose $\alpha_A$ and $\Pi_A$ that place target $q_j$ at row-index bit $j$\;
  $\widetilde{\Psi} \leftarrow \Pi_A(\Psi)$\;
  view $\widetilde{\Psi}$ as $M \in \mathbb{C}^{2^k\times 2^{N-k}}$\;
  $M' \leftarrow G M$\;
  $\widetilde{\Psi}' \leftarrow$ tensor view of $M'$\;
  $L' \leftarrow \alpha_A \circ L$\;
}
\Return $(L',\widetilde{\Psi}')$\;
\caption{Generic axis-shifted (lazy-TTGT) lowering of a fused $k$-qubit operator. $L$ maps logical qubits to current physical index-bit positions. The dense branch materialises the forward permutation $\Pi_A$ and records its resulting layout as $L'=\alpha_A\circ L$; the omitted data permutation $\Pi_A^{-1}$ and corresponding layout repair $\alpha_A^{-1}$ are deferred. The diagonal branch leaves the layout unchanged.}
\label{alg:axis-shift}
\end{algorithm}

For the generic lowering, the cost model separates explicit permutation traffic from arithmetic. We write $\mathcal{V}$ for a data-movement volume measured by the number of complex amplitudes transferred. When materialised, the permutation moves
\begin{equation}
\mathcal{V}_{\mathrm{perm}}(N) = \Theta(2^N)\;\text{complex amplitudes}
\label{eq:cost-perm}
\end{equation}
through memory (the whole state, once); the copy is elided when the targets already occupy the matrix-row positions in the required order. The GEMM performs
\begin{equation}
W_{\mathrm{dense}}(N,k) = \Theta(2^{N+k})\;\text{complex multiply-adds}.
\label{eq:cost-formal}
\end{equation}
For a diagonal fused operator, the permutation and GEMM collapse to a streaming elementwise multiply:
\begin{equation}
W_{\mathrm{diag}}(N) = \Theta(2^N).
\label{eq:cost-diag-formal}
\end{equation}
\qsim exposes these two scheduling terms above the backend API; CUDA, HIP, and CPU backends implement them with different libraries and constants. The two terms are profiled separately on GH200 (Appendix~\ref{appendix:resource-profile}); Appendix~\ref{appendix:fusion-sensitivity} lists the backend-specific fusion caps, including the Hopper setting evaluated in \S\ref{sec:eval-3way}.

The fuser, density-matrix wrapper, and rank-local dispatcher used during MPI execution submit fused operators with explicit target order; each backend then selects a lowering.

The density-matrix engine reuses this contract by representing an $N$-qubit density matrix $\rho\in\mathbb{C}^{2^N\times2^N}$ as a vector with $2N$ binary axes. Column-major vectorisation stacks the columns of $\rho$ into $\operatorname{vec}(\rho)\in\mathbb{C}^{2^{2N}}$: the lower $N$ bits of the vector index select a matrix row, and the upper $N$ bits select a matrix column. For a unitary $U\in\mathbb{C}^{2^N\times2^N}$, this convention gives
\[
\operatorname{vec}(U\rho U^\dagger) = (\overline{U}\otimes U)\operatorname{vec}(\rho),
\]
where $\overline{U}$ denotes entrywise conjugation. A channel with $R$ Kraus operators $E_i\in\mathbb{C}^{2^k\times2^k}$ is represented after vectorisation by $\sum_{i=1}^{R} \overline{E_i}\otimes E_i$. Each Kronecker term is a $4^k\times4^k$ matrix containing $16^k$ entries, so materialising and summing all $R$ terms costs $\Theta(R\cdot16^k)$. A logical $k$-qubit channel therefore acts on $2k$ axes in the doubled representation. For each ordered target tuple, the wrapper passes the row-copy positions first and the column-copy positions second, preserving target order within each group. The original factor acts on the lower row bits, and the conjugate factor acts on the more-significant column bits. The implementation constructs the dense superoperator once. Replacing $N$ by $2N$ and $k$ by $2k$ in Eq.~\eqref{eq:cost-formal} gives $\Theta(2^{2N+2k})$ application work; the vector and superoperator contain $4^N$ and $16^k$ entries, respectively.

\subsection{Fusion on the state-vector interface}\label{sec:methods-fusion}

Fusion replaces consecutive eligible gates by their matrix product. A wider fused block reduces the number of gate applications and layout changes, but its matrix grows with the number of distinct target qubits. \qsim first performs a semantics-preserving reorder, moving a gate only across disjoint or commuting neighbours, and then chooses the block boundaries with a deterministic dynamic program. Within a block, matrices remain in circuit order, with later gates multiplied on the left. Each operation carries a conservative marker for eligibility for the streaming diagonal lowering. A block has mode $\mathrm{diag}$ when every constituent operation carries that marker and mode $\mathrm{dense}$ otherwise.

For a block of width $k\ge1$, the scheduling score is
\begin{align*}
\gamma(k, \mathrm{dense}) &= 4^k \\
\gamma(k, \mathrm{diag})  &= 1 ,
\end{align*}
where $k$ is the number of distinct target qubits in the block. The dense score is the number of entries in its $2^k\times2^k$ matrix. A diagonal block receives unit score because it can use a streaming elementwise kernel. These scores guide the partition rather than predict its runtime.

Only unitary gates without attached noise enter the dynamic program. A barrier, measurement, reset, control-flow marker, or noise-bearing gate ends the current fusible run and passes through unchanged. A run longer than $\mathcal{W}=64$ gates is divided into consecutive optimisation windows of at most $\mathcal{W}$ gates, bounding the quadratic search. The evaluated configurations cap dense blocks at $K_{\max}$ and marked diagonal blocks at a backend-specific value no greater than $K_{\max}$. An input gate wider than its applicable cap is retained as a singleton rather than rejected. We therefore use $K_{\mathrm{blk}}$ for the larger of $K_{\max}$ and the width of the widest eligible input gate; it bounds the width of any emitted block in these configurations.

Consider one window containing $h\le\mathcal{W}$ reordered gates $g_0,\ldots,g_{h-1}$. Let $B(j,i)$ be the candidate block formed by the consecutive gates $g_j,\ldots,g_{i-1}$. Its width $k(B)$ is the number of distinct target qubits, and its block score is $\Gamma(B)=\gamma(k(B),\operatorname{mode}(B))$. Let $D_i$ be the minimum score for partitioning the first $i$ gates in the window, with $D_0=0$. If the final block begins at $j$, the preceding blocks contribute $D_j$ and the final block contributes $\Gamma(B(j,i))$. Let $\mathcal{J}_i$ contain every $j\in\{0,\ldots,i-1\}$ for which the width of $B(j,i)$ does not exceed the cap for its mode, and always include $i-1$ so that an input gate wider than its fusion cap remains an unfused singleton. Therefore
\begin{equation}
D_i = \min_{j\in\mathcal{J}_i}
\bigl\{D_j+\Gamma(B(j,i))\bigr\},
\qquad 1\le i\le h.
\label{eq:dp-recurrence}
\end{equation}
For each endpoint $i$, the recurrence tests every feasible start $j$ of the final block; backtracking from $D_h$ recovers the chosen boundaries. The search costs $\bigO{\mathcal{W}^2\cdot K_{\mathrm{blk}}}$ per window and $\bigO{|C|\cdot\mathcal{W}\cdot K_{\mathrm{blk}}}$ over a circuit with $|C|$ parsed instructions. Constructing a dense product for a multi-operation block of width $k\le K_{\max}$ by conventional matrix multiplication costs $\bigO{8^k}$, giving the conservative circuit-wide bound $\bigO{|C|\cdot8^{K_{\max}}}$ in Tab.~\ref{tab:complexity-summary}. An over-cap singleton reuses its input matrix and is not included in this construction term. After the windows are partitioned, a greedy pass attempts to merge blocks across their boundaries. The complete reorder-and-fuse pipeline is repeated at most four times and stops earlier when the gate count no longer decreases. \S\ref{sec:eval-3way} measures the choice of $K_{\max}$. A diagonal-aware DP score changes the measured Hopper intervals by less than $2\%$ (Appendix~\ref{appendix:fusion-sensitivity}); the diagonal gain comes from the streaming kernel rather than the partition score.

\subsection{Distributed logical-to-physical placement (R3)}\label{sec:methods-mpi}

A distributed layout can be viewed as a cache. Local index-bit positions are the cache slots, logical qubits are the items, and each fused gate requests its target qubits. A target currently stored in a rank-address bit causes a \emph{promotion}: the dispatcher exchanges that bit with an available local position. Retaining the new placement avoids an immediate reverse exchange and makes the future gate stream the request sequence for the paging model below.

A state vector distributed across $P=2^r$ ranks ($0\le r<N$) keeps $\Qloc=N-r$ of its $N$ index bits local to each rank; the remaining $r$ bits are rank-address bits. Let $\ket{\psi}$ denote the global state and $\psi_x=\langle x\vert\psi\rangle$ its amplitude at flat basis index $0\le x<2^N$. For rank $p\in\{0,\ldots,P-1\}$, let $\boldsymbol{\psi}^{(p)}\in\mathbb{C}^{2^{\Qloc}}$ be its local amplitude vector and $\psi^{(p)}_j$ its entry at local offset $j$. A canonical block distribution assigns
\begin{equation}
\psi^{(p)}_j=\psi_{p\cdot2^{\Qloc}+j},\qquad 0\le j<2^{\Qloc}.
\label{eq:mpi-block}
\end{equation}
Thus rank $p$ owns a contiguous range of global amplitudes. The $r$ high physical index bits hold the rank address $p$, and the $\Qloc$ low bits hold the local offset $j$. A gate whose target maps to a rank-address bit requires data movement before a local kernel can apply it. Subsequent gates can retain that placement as long as dispatch and output use the updated map.

The dispatcher stores a bijection from logical qubits to physical index bits,
\[
\pi_t:\{0,\ldots,N-1\}\rightarrow\{0,\ldots,N-1\},
\]
together with its inverse. The step index $t$ advances on every promotion and fused-operator application.

To specify the direction of the map, write a basis index $x$ in bits $x_0,\ldots,x_{N-1}$, with $x_0$ least significant. Let $y$ be the bit string obtained by placing logical bit $x_q$ at physical position $\pi_t(q)$, so that $y_{\pi_t(q)}=x_q$ for every logical qubit $q$. The corresponding permutation operator is defined on computational-basis states by
\[
\Pi_{\pi_t}\ket{x}=\ket{y}.
\]
Thus $\Pi_{\pi_t}$ carries logical bit $q$ to physical position $\pi_t(q)$. For a fused operator on $V=(q_0,\ldots,q_{k-1})$, the ordered physical targets are $A_t=\pi_t(V)=(\pi_t(q_0),\ldots,\pi_t(q_{k-1}))$.

If a target $b\in A_t$ is a rank-address bit, the dispatcher exchanges it with an unpinned local index bit $\ell$. If $s$ of the $k$ targets are rank-address bits, then $k-s$ local targets are pinned. Whenever $k\le\Qloc$ (the hypothesis of Theorem~\ref{thm:lazy-mpi-invariant} below), the number of available local bits is $\Qloc-(k-s)\ge s$, so all promotions are feasible.

Write $\ket{\psi_t}$ for the canonical logical state after step $t$, the state a canonically ordered execution would hold, and $\ket{\widetilde{\psi}_t}$ for the state as physically stored. The tilde marks permuted storage order, exactly as for the axis-permuted tensor $\widetilde{\Psi}$ of \S\ref{sec:methods-axisshift}. Equation~\eqref{eq:mpi-block} applies to the amplitude vector of $\ket{\widetilde{\psi}_t}$.

\begin{theorem}[Distributed placement invariant]\label{thm:lazy-mpi-invariant}
Assume an initial bijection $\pi_0$ and stored state satisfying $\ket{\widetilde{\psi}_0}=\Pi_{\pi_0}\ket{\psi_0}$, and assume every fused operator acts on at most $\Qloc$ qubits. After each step $t$, the stored and canonical states satisfy
\begin{equation}
  \ket{\widetilde{\psi}_t}
    =\Pi_{\pi_t}\ket{\psi_t}.
\label{eq:lazy-mpi-invariant}
\end{equation}
\end{theorem}

\begin{proof}[Proof sketch]
The base case is the assumed initial invariant. A distributed bit transpose swaps two physical index bits; composing $\pi_t$ with the same transposition preserves Eq.~\eqref{eq:lazy-mpi-invariant}. Once all targets are local, write their ordered physical positions as $A=\pi_t(V)$. Applying the physical fused operator to $A$ is equivalent to applying the logical operator to $V$; a lowering that retains the local axis relabelling $\alpha_A$ (recording $L'=\alpha_A\circ L$, \S\ref{sec:methods-axisshift}) is absorbed the same way a promotion is. Let $\hat{\alpha}_A$ denote the extension of $\alpha_A$ by the identity on rank-address bits; updating $\pi_{t+1}=\hat{\alpha}_A\circ\pi_t$ then records the retained relabelling. Induction over promotions and fused operators proves the invariant. Appendix~\ref{appendix:lazy-proof} gives the component-wise form.
\end{proof}

The production planner may choose a nonidentity $\pi_0$ before execution. A metadata-only relabel is valid for the permutation-invariant initial state $\ket{0^N}$ used here; an arbitrary supplied state must first be physically reordered or otherwise made to satisfy the theorem's initial invariant.

A consumer may either decode its result through $\pi_t$ or request a physical layout compatible with its operation. In particular, terminal computational-basis count sampling decodes logical output bits through the final map and does not require canonical amplitude order.

\begin{algorithm}
\SetAlgoLined
\KwIn{Fused operator $G$ on $V=(q_0,\ldots,q_{k-1})$; the current map $\pi=\pi_t$ and its inverse; local count $\Qloc$.}
\KwOut{$G$ applied to the distributed state; maps updated.}
\tcp{Precondition: $k\le\Qloc$.}
$A\leftarrow(\pi(q_0),\ldots,\pi(q_{k-1}))$\;
$\mathit{Pin}\leftarrow\{b\in A:b<\Qloc\}$ \tcp*{already-local target positions, pinned}
\ForEach{$b\in A$ with $b\ge\Qloc$}{
  choose $\ell\in\{0,\ldots,\Qloc-1\}\setminus \mathit{Pin}$ by the placement policy\;
  perform the selected distributed lowering for $(b\leftrightarrow\ell)$\;
  update $\pi$ and $\pi^{-1}$ by exchanging physical positions $b$ and $\ell$\;
  replace $b$ by $\ell$ in $A$ and insert $\ell$ into $\mathit{Pin}$\;
}
apply $G$ to ordered local positions $A$ through the R2 interface\;
\If{the selected R2 lowering changes the order of local bit positions}{
  compose $\pi$ with that returned position permutation, leaving rank-address positions fixed; rebuild $\pi^{-1}$\;
}
\caption{Logical-to-physical MPI dispatch. Placement metadata retains the resulting non-canonical order. A later consumer may decode the map or request materialisation.}
\label{alg:lazy-mpi}
\end{algorithm}

Algorithm~\ref{alg:lazy-mpi} lists this dispatch. Once its targets are local, a dense $k$-qubit block costs $\Theta(2^{\Qloc+k})$ arithmetic per rank, obtained by replacing $N$ with $\Qloc$ in Eq.~\eqref{eq:cost-formal}; structured blocks may use cheaper lowerings. The evaluated configurations use the backend-specific fusion caps of Appendix~\ref{appendix:fusion-sensitivity} and satisfy the distributed chunk-width constraint of Appendix~\ref{appendix:int-max}.

The following model counts promotions one target position at a time. One promotion exchanges a rank-address target with a local position and sends $\nu=2^{\Qloc-1}$ amplitudes per rank. Across all ranks this is $\Theta(2^N)$ amplitudes, up to transport-dependent constants.

Suppose the backend receives fused blocks $G_1,\ldots,G_J$. In the \emph{eager} baseline, let $p_i^{\mathrm E}$ be the number of forward promotions used to make the targets of $G_i$ local. Immediately after applying $G_i$, the baseline reverses those exchanges, so the block ends in the layout in which it began. In the \emph{lazy} execution, let $p_i^{\mathrm L}$ be the number of promotions required by $G_i$ from the layout retained after $G_{i-1}$, and let $p_{\mathrm{out}}$ count any additional cross-boundary exchanges demanded by the final consumer. The corresponding per-rank communication volumes, measured in complex amplitudes, are
\begin{equation}
\begin{aligned}
\mathcal{V}_{\mathrm{eager}}
  &=2\nu\sum_{i=1}^{J}p_i^{\mathrm E},\\
\mathcal{V}_{\mathrm{lazy}}
  &=\nu\left(\sum_{i=1}^{J}p_i^{\mathrm L}+p_{\mathrm{out}}\right).
\end{aligned}
\label{eq:lazy-vol}
\end{equation}
The factor two in $\mathcal{V}_{\mathrm{eager}}$ is the forward-and-reverse round trip. Lazy execution does not pay that return immediately: the promoted target remains local and may make a later $p_i^{\mathrm L}$ zero. The two sequences of counts need not be equal because they are evaluated under different evolving layouts. For permutation-aware terminal count sampling, the output is decoded through the retained map and $p_{\mathrm{out}}=0$.

\begin{corollary}[Ideal minimal demand-promotion count]\label{cor:belady}
Treat the $\Qloc$ local positions as page frames and the logical qubits as pages. Fix the fused stream and initial frame contents, with distinct targets in every tuple and at most $\Qloc$ targets per tuple. A missing target is promoted only on demand, at uniform cost, by evicting one local bit; any bit not targeted by the current tuple is eligible. Let $p_i$ be the number of promotions made for target tuple $i$. The policy knows the full future and does not prefetch. Call this the \emph{all-frame demand model}. Within this model, evicting the eligible bit whose next use is farthest in the future minimises $\sum_{i=1}^{J}p_i$.
\end{corollary}

\begin{proof}[Justification]
Let $V_i$ denote the ordered target tuple of the $i$-th fused operator. Expand each tuple into its constituent qubit references, write that list twice, and concatenate the resulting lists. The block structure of this ordinary scalar reference string is $V_1,V_1,V_2,V_2,\ldots$. Every reference in the first copy of a block has its next use inside that block, whereas a frame not targeted by the block is next used no earlier than the following block; farthest-next-use therefore never evicts a current target while the first copy is being served, and the second copy is all hits. Doubling also preserves the next-use order of all unpinned frames. Conversely, a pin-respecting batched policy can pay its batch misses while serving the first copy and then serve the second copy entirely from the resulting cache, so its batched promotion count equals its fault count on the doubled string. The batched pinned problem therefore has the same minimum as ordinary offline paging on the doubled string. Farthest-next-use eviction, introduced by B\'el\'ady~\cite{belady1966study} and proved optimal in Mattson et al.'s framework~\cite{mattson1970evaluation}, minimises the fault count on that string.
\end{proof}

The model minimises promotion count only: it prices neither message shape, topology, protocol overhead, nor overlap. The production GPU-MPI planner approximates farthest-next-use over a 1024-entry schedule window and restricts victims to the first $\Qloc-1$ local positions. CPU-MPI uses \emph{first-free} placement, the lowest-indexed unpinned local position. On the strong-scaling fused stream, the GPU window sees all remaining references, while the restricted victim set still separates the implementation from the corollary; \S\ref{sec:eval-placement} quantifies the policies.

Placement selects which axes to promote; chunked dispatch moves them through intra-rank index swaps and XOR-partner or coalesced cross-rank transfers. Diagonal and controlled operations retain their fast paths, and a full-permutation fallback may use \code{MPI\_Alltoallv}.

\subsection{Backend lowering through the shared primitive contract (R4)}\label{sec:design-divergences}

After the router selects a method and the layout layer resolves its targets, the backend executes the four backend primitives. Tab.~\ref{tab:backend-matrix} lists the implementations used by the common state-vector path.

\begin{table*}[t]
\caption{Representative lowerings and primitives for the full-state path. Backends preserve fused-operator semantics and logical layout but may choose different specialised kernels, libraries, arithmetic modes, and transports.}
\label{tab:backend-matrix}
\centering
\footnotesize \setlength{\tabcolsep}{3pt}
\renewcommand{\arraystretch}{1.12}
\begin{tabular}{@{}>{\raggedright\arraybackslash}p{0.16\textwidth}
                >{\raggedright\arraybackslash}p{0.25\textwidth}
                >{\raggedright\arraybackslash}p{0.26\textwidth}
                >{\raggedright\arraybackslash}p{0.25\textwidth}@{}}
\toprule
Operation & CUDA / NVIDIA & HIP / AMD & CPU \\
\midrule
Target-axis permutation &
\texttt{cuTensor} permutation &
Custom complex-valued permutation kernel; the evaluated ROCm stack lacks the
required \texttt{hipTensor} datatype path &
Host permutation \\
\addlinespace
Fused update &
Diagonal/controlled kernels, small-width kernels, or \texttt{cuBLAS} fallback &
Diagonal/controlled kernels, custom dense kernels, or \texttt{hipBLAS}/\texttt{hipBLASLt} fallback &
In-place diagonal and small-width kernels; host permutation with OpenBLAS or NVPL GEMM fallback \\
\addlinespace
Diagonal updates and sampling &
Device kernels; single-GPU two-level sampler with $2^{10}$-entry blocks and
FP64 block totals &
Device kernels; the same single-GPU sampler &
Host loops; FP64 cumulative distribution with $2^{16}$-entry blocks in the local CPU state-vector engine \\
\addlinespace
Distributed bit transpose &
XOR-partner and coalesced exchanges through HPC-X OpenMPI/UCX &
XOR-partner and coalesced exchanges through Cray MPICH and its Heterogeneous System Architecture (HSA) GPU Transfer Library (GTL) &
First-free placement with pairwise or full-permutation MPI lowerings \\
\bottomrule
\end{tabular}
\end{table*}

One visible backend difference is target-axis permutation. CUDA uses \texttt{cuTensor}, whereas the evaluated ROCm 7.2.3 stack requires a custom complex-valued HIP kernel because \code{hiptensorPermute} returns \code{NOT\_SUPPORTED} for the required datatype. GEMM precision, sampling-kernel choice, and MPI transport likewise vary below the common execution path.

\paragraph{Sampling}

For a state $\ket{\psi}=\sum_{x=0}^{2^N-1}\psi_x\ket{x}$, state-vector engines draw computational-basis outcome $x$ with probability $p(x)=|\psi_x|^2$ through a hierarchical cumulative distribution function (CDF). CUDA, HIP, CPU, and MPI paths use different block decompositions, but retain block totals and inter-rank offsets in FP64; the MPI-GPU within-block prefixes remain FP32 to bound memory. Inverse-CDF search uses a strict upper bound, which prevents selection of a zero-probability bin. Appendix~\ref{appendix:reproducibility} records the concrete decompositions and numerical precision used in the evaluated builds.

\paragraph{CPU lowerings}

The CPU state-vector engine can apply a narrow fused block directly to amplitude groups instead of first permuting its target axes. Fixing the $N-k$ non-target bits identifies one independent group of $2^k$ amplitudes. A direct kernel gathers this group in the ordered-target convention of R2, computes the matrix--vector product with $G$, and writes the result back. Independence between groups permits OpenMP parallelism. Within a group or across adjacent groups, the same short sequence of complex multiply-adds recurs, making it suitable for single-instruction, multiple-data (SIMD) execution. Explicit NEON kernels cover one- and two-qubit updates, AVX2 covers one-qubit updates, and the compile-time matrix dimensions for $k=3,4,5$ allow the compiler to unroll and vectorise the loops; scalar implementations remain available.

For a narrow block, a target-axis permutation, an auxiliary state buffer, and a general matrix multiplication call can cost more than the local arithmetic. The dispatcher therefore streams diagonal blocks in place, uses direct dense kernels for $k\le3$ and for four- and five-qubit blocks beyond a state-size crossover, and otherwise uses host permutation followed by complex GEMM. Direct kernels assume the canonical flat-index layout, where qubit $q$ occupies index bit $q$, so the engine first materialises any pending lazy permutation. The generic branch may instead leave its new order represented in $L$. Both branches implement the same R2 semantics. Appendix~\ref{appendix:divergences} gives the concrete crossover used in the evaluated builds.

The tableau's bitwise operations on 64-bit integer arrays have a separate instruction-set dispatch across AVX-512, AVX2, SVE/SVE2, NEON, and scalar implementations. CPU-MPI applies the same local state-vector dispatch after R3 brings the targets into local index bits; its first-free victim policy remains the one defined in \S\ref{sec:methods-mpi}.

\subsection{Complexity summary}

Representation determines the exponential state size, while fusion and placement determine arithmetic intensity and data movement. Tab.~\ref{tab:complexity-summary} collects the bounds derived above.

Bounded-support fusion, GEMM-based gate application, tableau updates, and high-bit/low-bit state partitioning are established techniques~\cite{smelyanskiy2016qhipster,haner2017half,zhang2021hyquas,aaronson2004improved,juqcs2018}. \qsim combines them behind representation and layout contracts. The distinct distributed contribution is the explicit placement invariant and the offline-paging reduction under the model of Corollary~\ref{cor:belady}. \S\ref{sec:related} compares this design with UniQ, Atlas, and Lazy Qubit Reordering.

\section{Evaluation}\label{sec:eval}

The evaluation follows the architecture across three changes of regime: semantic routing between tableau and full-state execution, the fused-operator contract across NVIDIA and AMD backends, and explicit layout from local to distributed execution. Tab.~\ref{tab:eval-glance} collects the main results; the subsections that follow give the workloads and measurement details.

\begin{table*}[t]
\caption{Summary of the reported intervals (baseline short names: \S\ref{sec:eval-setup}). Appendix~\ref{appendix:reproducibility} defines the timing and repetition policies.}
\label{tab:eval-glance}
\centering\footnotesize
\begin{tabular}{@{}>{\raggedright\arraybackslash}p{0.18\textwidth}>{\raggedright\arraybackslash}p{0.14\textwidth}>{\raggedright\arraybackslash}p{0.60\textwidth}@{}}
\toprule
Target & Primary baseline & Main result \\
\midrule
Synthetic routing & Stim & At $N=30$, routing cuts \qsim's noisy-GHZ interval from 454.8~s to 5.8~ms and the tableau reaches $N=1000$; Stim stays ahead in every cell \\
\ella single device & \aer-noCuQ & On random brickwork, \qsim records the lowest Hopper interval at $N=30$ and 32 and the lowest Grace FP64 interval at $N=28$ and 30; \qsimg-CPU is lower on all four Grace cells \\
\setonix single device & \aer-GPU & On random brickwork, \qsim leads the four-engine MI250X comparison at $N=24$--30; on EPYC it leads the FP64 engines at $N=24$, while \aer leads FP64 at $N=26$--30 and \qsimg-CPU leads overall \\
\ella strong scaling & \aer-noCuQ & \aer-noCuQ takes the GPU lead from $P=4$, and QuEST the CPU lead from $P=4$ \\
\setonix strong scaling & \aer-GPU & \qsim is ahead of both comparators at every measured GPU and CPU rank \\
256-rank GPU weak scaling & \aer-GPU & Both reach $N=38$ (2~TiB); \qsim leads through $P=32$ and \aer from $P=64$ \\
256-rank CPU weak scaling & \aer-CPU & \qsim reaches $N=36$ (1~TiB), stays ahead of \aer throughout, and leads QuEST from $P=4$ (measured to $P=128$) \\
\bottomrule
\end{tabular}
\end{table*}

\subsection{Experimental setup}\label{sec:eval-setup}

\paragraph{Platforms} \ella is a Pawsey GH200 system with one Grace CPU and one Hopper GPU per node, ConnectX-7 RoCEv2, and HPC-X OpenMPI. \setonix combines MI250X GPU nodes with HPE Slingshot~11 and Cray MPICH; each node exposes eight GPU compute dies (GCDs). Its CPU-only partition uses EPYC 7763 nodes. Appendix~\ref{appendix:artifact} records the hardware and software environment.

\paragraph{Baselines} The state-vector comparison includes Qiskit \aer~\cite{aleksandrowicz2019qiskit}, \qsimg~\cite{isakov2021qsim}, Qulacs~\cite{suzuki2021qulacs}, PennyLane \lightning~\cite{asadi2024lightning}, and QuEST~\cite{jones2019quest}; the Clifford experiments add \aer's stabiliser method and Stim~\cite{gidney2021stim}. \aer-noCuQ is \aer's open GPU backend without \custate, \aer-cuQ links \aer to \custate, and \aer-GPU on \setonix is the ROCm/HIP build.

All tools receive the same raw \openqasm inputs without external transpilation. \aer, \lightning, and \qsimg parse them through framework front ends, Qulacs uses its converter, and QuEST uses the benchmark bridge. The \qsimg-GPU rows report the tool-native sampled-run interval, including state download and sampling.

\paragraph{Workloads} State-vector sweeps use 50-layer random brickwork, with alternating layers of random single-qubit rotations and nearest-neighbour two-qubit gates, \code{rand\{24,26,28,30,32\}x50}, and the quantum Fourier transform (QFT), \code{qft\{24,26,28,30,32\}}. The \setonix paired sweep additionally covers $N=18$--$22$ below the range shared by all four GPU tools. The MI250X sweep ends at $N=30$; \code{rand32x50} is reported separately as a managed-memory capacity case. The audited gate counts also show \code{rand26x50} to be a lighter instance (1{,}951 operations against 2{,}410--3{,}010 at $N=24$, 28, and 30), so per-size readings at $N=26$ partly reflect the instance. Synthetic GHZ and random-Clifford circuits exercise semantic routing. Application tests cover the quantum approximate optimisation algorithm (QAOA)~\cite{farhi2014qaoa}, unitary coupled-cluster singles-and-doubles (UCCSD)-style Pauli rotations~\cite{peruzzo2014variational}, and QASMBench~\cite{li2023qasmbench}.

\paragraph{Comparison protocol} Except for the application study, the cross-tool figures use each simulator's internal execution interval. \qsim includes gate execution and sampling; \aer times \code{.run()} after construction; QuEST times its gate-evolution loop, excluding environment setup and sampling, and every QuEST interval below uses that boundary. The application study uses its separately predefined runner boundaries (\S\ref{sec:eval-applications}). The paired \ella single-GPU and \setonix single-GCD comparisons use one exclusive allocation, vary the tool order between repetitions, and report five-run medians. Table captions identify medians and individual observations for the other series; Appendix~\ref{appendix:reproducibility} consolidates the timing, shot-count, repetition, and reduction policies.

\subsection{Semantic routing on synthetic Clifford workloads}\label{sec:eval-stab}

Clifford circuits can use a polynomial-size tableau instead of an exponential state vector. We submit noisy GHZ and random-Clifford inputs through the generic \code{gpu} entry point and test whether \qsim routes them before GPU allocation. The GHZ sweep attaches depolarising noise with $p_{\mathrm{dep}}=10^{-2}$ after the Hadamard and each CX gate, while the random-Clifford inputs are noiseless. Using one common rate is a benchmark choice rather than a routing requirement: supported Pauli-channel parameters may vary by circuit location. Timings are reported on \ella; the routing decision depends on the circuit, noise model, and requested output semantics rather than the hardware backend.

\begin{figure}[t]
\centering
\includegraphics[width=\linewidth]{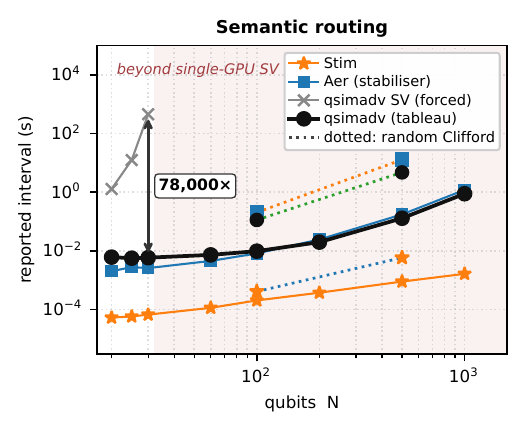}
\caption{Synthetic Clifford routing on \ella. Colour and marker shape identify the simulator; solid curves show noisy GHZ inputs and dotted curves show random-Clifford inputs. The forced single-GPU state-vector path ends at $N=30$, whereas automatic tableau routing reaches $N=1000$. Values are listed in Tab.~\ref{tab:qec-stab}.}
\label{fig:qec}
\end{figure}

Semantic routing changes the feasible problem size rather than merely shortening a fixed run (Fig.~\ref{fig:qec}). The forced single-GPU state-vector path ends at $N=30$, where it already takes 454.8~s. The same generic \code{gpu} request takes 5.8~ms after routing and proceeds to $N=1000$ in 0.884~s, far beyond the memory limit of full-state allocation, without requiring the user to choose a stabiliser engine. The routed path also overtakes \aer's stabiliser sampler beyond approximately 200 qubits on GHZ and at both random-Clifford sizes, by up to $2.7\times$ on \code{randcliff500}. Stim remains fastest in every measured cell.

As an output sanity check, the three tools' all-zero GHZ frequencies differ by at most 1.02 independent-binomial standard errors. This is one sampled marginal, not complete state fidelity or a check of the random-Clifford distributions (Appendix~\ref{appendix:qec-detail}).

\subsection{The fused-operator interface across backends}\label{sec:eval-3way}

The same fused-operator contract drives four backends: Hopper GPU and Grace CPU on \ella, and MI250X GPU and EPYC CPU on \setonix. Each backend supplies its own permutation, matrix, and structured-gate kernels (Tab.~\ref{tab:backend-matrix}). The random-brickwork comparisons occupy panels (a) and (b) of Figs.~\ref{fig:ella} and~\ref{fig:setonix}; the strong-scaling panels (c) and (d) are analysed separately in \S\ref{sec:eval-mpi}. QFT results follow in \S\ref{sec:eval-qft}.

\begin{figure*}[t]
\centering
\includegraphics[width=\textwidth]{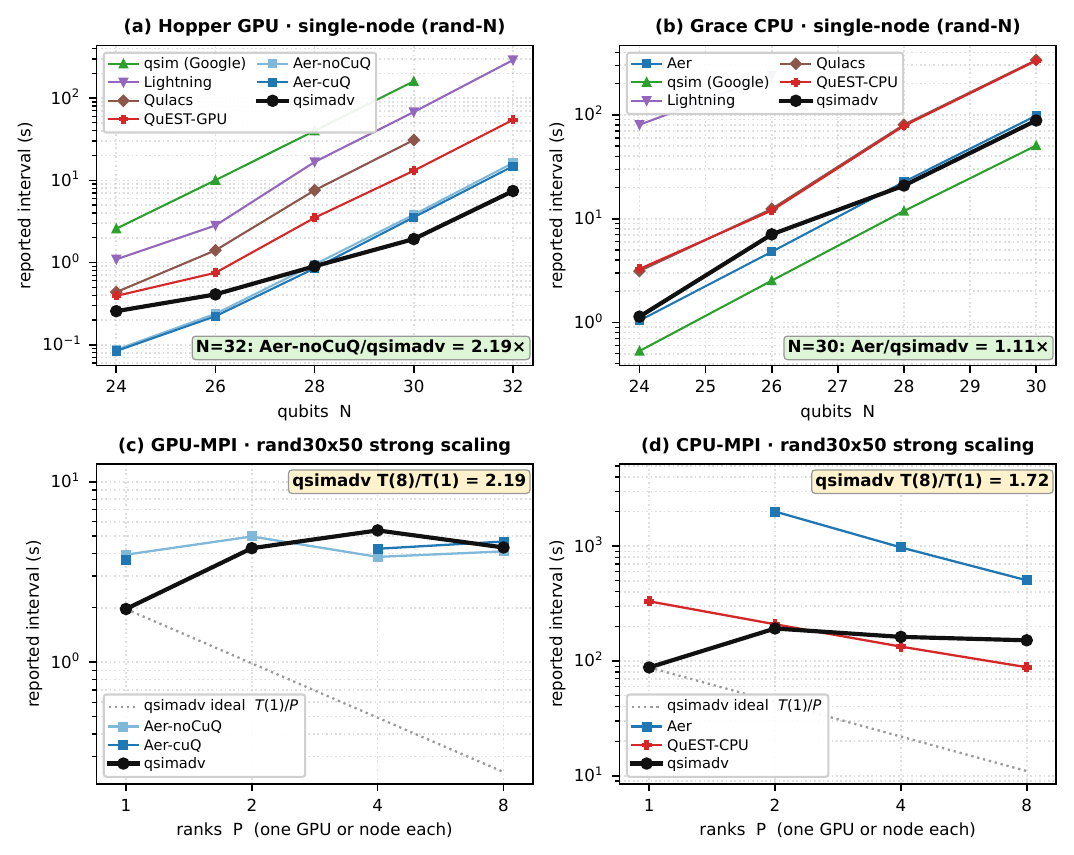}
\caption{\ella execution intervals: random-brickwork single-node results on (a) Hopper GPU and (b) Grace CPU, followed by \code{rand30x50} strong scaling on (c) GPU-MPI and (d) CPU-MPI. Here $T(P)$ denotes \qsim's interval on $P$ ranks; the dotted $T(1)/P$ reference is ideal strong scaling. \lightninggpu and Qulacs-GPU use FP64; other GPU series use FP32.}
\label{fig:ella}
\end{figure*}

\begin{figure*}[t]
\centering
\includegraphics[width=\textwidth]{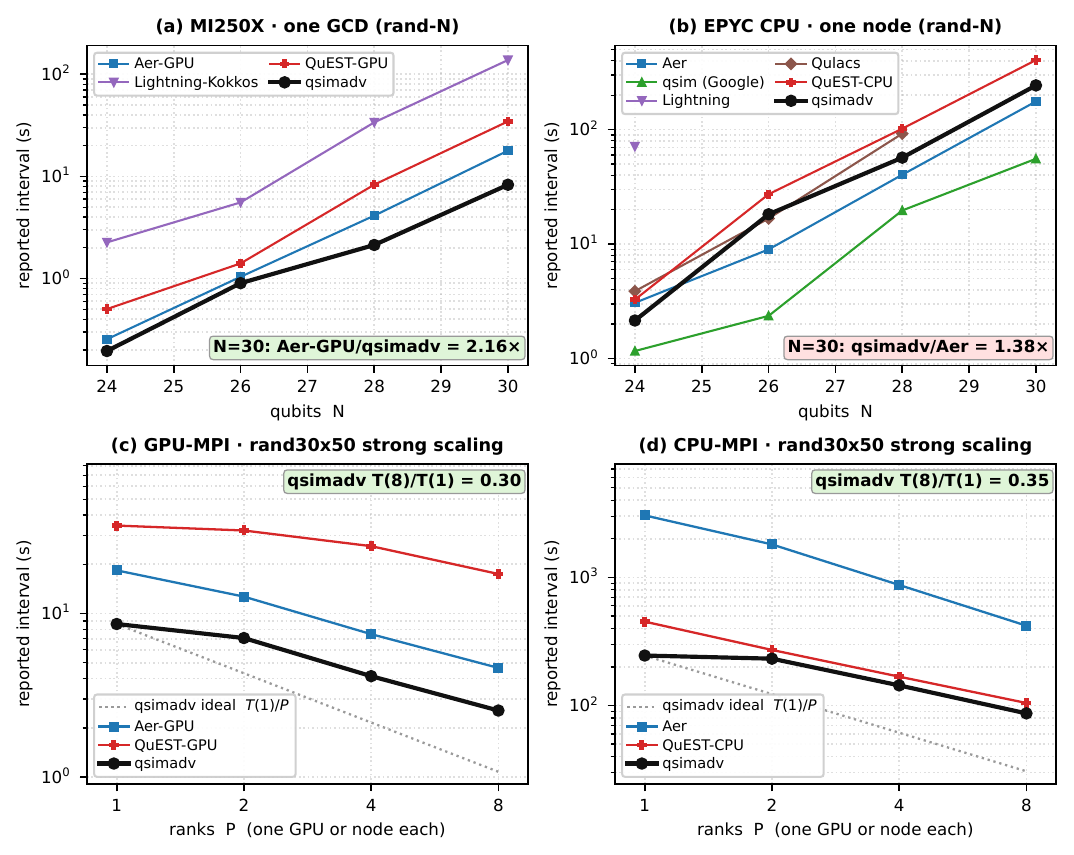}
\caption{\setonix execution intervals: random-brickwork results on (a) one MI250X GCD and (b) one EPYC CPU node, followed by \code{rand30x50} strong scaling on (c) GPU-MPI and (d) CPU-MPI. Here $T(P)$ denotes \qsim's interval on $P$ ranks; the dotted $T(1)/P$ reference is ideal strong scaling. \lightningkokkos uses FP64; other GPU series use FP32.}
\label{fig:setonix}
\end{figure*}

\paragraph{Hopper GPU} Fig.~\ref{fig:ella}(a) and Tab.~\ref{tab:3way-gpu} compare seven GPU paths. At $N=24$ and 26, both \aer variants are faster than \qsim, while \qsim is already ahead of \qsimg-GPU, \lightninggpu, Qulacs-GPU, and QuEST-GPU. At $N=28$, \qsim overtakes the open \aer-noCuQ path but remains 6\% behind the \custate-backed \aer-cuQ path. It records the lowest interval in the panel at $N=30$ and 32: the \aer-noCuQ-to-\qsim ratios are 1.98 and 2.19, and the \aer-cuQ-to-\qsim ratios are 1.84 and 2.00. QuEST-GPU, the next non-\aer baseline at these sizes, takes 6.8 and 7.3 times the \qsim interval. The five paired observations in each \qsim/\aer cell span at most $2.2\%$ of the median, except at $N=24$, where the span is $3.9\%$.

The Hopper evaluation also includes a fusion-cap sweep on \code{rand32x50}, the largest random circuit. Wider fusion reduces the number of full-state updates, but increases the fused-matrix size and construction cost. One run at each of $K_{\max}=5,6,7$ gives intervals of 12.35, 8.51, and 7.39~s, respectively. At this state size the sweep favours $K_{\max}=7$, the value used by the CUDA path; Appendix~\ref{appendix:fusion-sensitivity} gives the full size-dependent and backend rules.

The precision bridge in Tab.~\ref{tab:precision-bridge} adds a separate FP64 \aer sweep to the paired Hopper data. Its \aer-to-\qsim time ratios on \code{rand24x50}, \code{rand26x50}, and \code{rand28x50} are 0.49, 1.01, and 1.89, respectively.

\paragraph{Grace CPU} Fig.~\ref{fig:ella}(b) and Tab.~\ref{tab:3way-cpu} show a second large-$N$ crossover. \qsim trails \aer at $N=24$ and 26, overtakes it at $N=28$, and retains the lowest FP64 interval at $N=30$. \lightning, Qulacs, and QuEST are slower than \qsim at all four plotted widths. The mixed-precision ordering is different: \qsimg-CPU uses FP32 and records the lowest interval in every cell.

\begin{table*}[!t]
\caption{Single-GPU intervals on \ella. Values are intervals in seconds; bold marks the lowest available value in each row. The three leftmost columns are five-run paired medians; the other four are single observations. ``T/O'' denotes timeout; ``OOM'' marks Qulacs-GPU cells whose two-buffer footprint exceeds device memory.}
\label{tab:3way-gpu}
\centering\footnotesize
\resizebox{\textwidth}{!}{%
\begin{tabular}{@{}lrrrrrrr@{}}
\toprule
Circuit       & \qsim & \aer-noCuQ & \aer-cuQ & \lightninggpu & \qsimg-GPU & Qulacs-GPU & QuEST-GPU \\
\midrule
\code{rand24x50}  & 0.256 & 0.087 & \good{0.084} & 1.095 & 2.581 & 0.437 & 0.393 \\
\code{rand26x50}  & 0.410 & 0.238 & \good{0.222} & 2.825 & 10.04 & 1.409 & 0.751 \\
\code{rand28x50}  & 0.901 & 0.949 & \good{0.850} & 16.72 & 40.06 & 7.565 & 3.506 \\
\code{rand30x50}  & \good{1.932} & 3.821 & 3.554 & 67.82 & 160.0 & 30.96 & 13.10 \\
\code{rand32x50}  & \good{7.415} & 16.27 & 14.85 & 290.4 & T/O & OOM & 54.40 \\
\code{qft24}      & 0.231 & 0.020 & \good{0.018} & 0.054 & 2.479 & 0.132 & 0.046 \\
\code{qft26}      & 0.267 & 0.068 & \good{0.052} & 0.173 & 9.885 & 0.552 & 0.104 \\
\code{qft28}      & 0.419 & 0.277 & \good{0.201} & 0.668 & 39.52 & 2.255 & 0.347 \\
\code{qft30}      & 0.930 & 1.221 & \good{0.861} & 2.853 & 158.1 & 9.510 & 1.338 \\
\code{qft32}      & \good{3.744} & 5.398 & 3.817 & 12.62 & T/O & OOM & 5.705 \\
\bottomrule
\end{tabular}}
\end{table*}

\paragraph{MI250X GPU} Fig.~\ref{fig:setonix}(a) and Tab.~\ref{tab:setonix-engines} show \aer-GPU ahead at $N=18$--22, followed by a crossover between $N=22$ and 24. From $N=24$ through 30, \qsim records the lowest interval among all four GPU engines. The \aer-to-\qsim ratio ranges from 1.15 to 2.16 over those shared sizes. At $N=30$, QuEST-GPU takes 4.2 times the \qsim interval and the FP64 \lightningkokkos path takes 16.6 times the interval.

\begin{table}[!htp]
\caption{Single-GCD intervals on \setonix. Values are intervals in seconds; bold marks the lowest available value in each row. \qsim and \aer-GPU are five-run paired medians; QuEST-GPU and \lightningkokkos are single observations. \lightningkokkos uses FP64 and the other columns FP32. Rows below $N=24$ were not measured by every tool; ``N/M'' denotes an unmeasured cell.}
\label{tab:setonix-engines}
\centering\small
\resizebox{\columnwidth}{!}{%
\begin{tabular}{lrrrr}
\toprule
Circuit          & \qsim & \aer-GPU & QuEST-GPU & \lightningkokkos \\
\midrule
\code{rand18x50} & 0.093 & \good{0.027} & N/M & 0.249 \\
\code{rand20x50} & 0.098 & \good{0.038} & N/M & 0.317 \\
\code{rand22x50} & 0.121 & \good{0.080} & N/M & 0.565 \\
\code{rand24x50} & \good{0.195} & 0.254 & 0.502 & 2.259 \\
\code{rand26x50} & \good{0.901} & 1.040 & 1.406 & 5.562 \\
\code{rand28x50} & \good{2.131} & 4.130 & 8.302 & 33.82 \\
\code{rand30x50} & \good{8.283} & 17.86 & 34.56 & 137.8 \\
\code{qft18}     & 0.081 & \good{0.006} & N/M & 0.033 \\
\code{qft20}     & 0.082 & \good{0.007} & N/M & 0.039 \\
\code{qft22}     & 0.086 & \good{0.013} & N/M & 0.055 \\
\code{qft24}     & 0.105 & \good{0.037} & 0.046 & 0.104 \\
\code{qft26}     & 0.194 & \good{0.140} & 0.178 & 0.311 \\
\code{qft28}     & \good{0.594} & 0.604 & 0.776 & 1.225 \\
\code{qft30}     & \good{2.576} & 2.649 & 3.488 & 5.244 \\
\bottomrule
\end{tabular}}
\\[2pt]{\footnotesize The \code{rand32x50} capacity cell is discussed in Appendix~\ref{appendix:resource-profile}.}
\end{table}

\paragraph{EPYC CPU} Fig.~\ref{fig:setonix}(b) and Tab.~\ref{tab:setonix-cpu-3way} complete the four-backend comparison. At $N=24$, \qsim has the lowest FP64 interval. \aer takes the FP64 lead from $N=26$ and is $1.38\times$ faster than \qsim at $N=30$. Qulacs also edges \qsim at $N=26$, but \qsim is the next FP64 engine after \aer at $N=28$ and 30 and remains ahead of QuEST at every plotted width. As on Grace, the FP32 \qsimg-CPU path records the lowest interval throughout the panel.

Across both GPU vendors, the crossover moves in \qsim's favour as $N$ grows, and \qsim leads the larger random-brickwork cells. The CPU ordering is less uniform: \qsim gains the large-$N$ FP64 lead on Grace but not on EPYC, while the FP32 \qsimg-CPU path remains fastest on both.

\subsection{Diagonal-heavy QFT}\label{sec:eval-qft}

QFT interleaves Hadamards with controlled phases and favours element-wise diagonal updates over generic TTGT. Fig.~\ref{fig:qft} plots the $N\ge24$ range shared across the systems. On Hopper, the crossing sits one sweep step above the random-brickwork one: the \aer-noCuQ-to-\qsim time ratio rises from 0.09 at \code{qft24} to 1.31 at \code{qft30} and 1.44 at \code{qft32}, while the corresponding \aer-cuQ ratios are 0.93 and 1.02 at the two larger sizes. On MI250X, \aer-GPU is $2.9\times$ and $1.4\times$ faster at $N=24$ and $N=26$, and \qsim takes the lower median at both larger sizes, where the paired medians differ by only 1.7\% and 2.8\%. \aer has the lowest FP64 time on every CPU QFT cell of both machines.

\subsection{Application circuits}\label{sec:eval-applications}

The non-Clifford application inputs exercise the full-state R2 path, while the Clifford \code{qec9xz\_n17} input exercises R1 on a public benchmark. The sweep reports five-run medians under its predefined runner policy: \qsim is timed over its whole process, while \aer and \lightning are timed after construction (Appendix~\ref{appendix:reproducibility}). The \qsim interval therefore includes stages excluded from the two baseline intervals. \lightning or \aer reports the lower interval on the smaller non-Clifford QAOA, UCCSD-style, and QASMBench inputs. At $N=30$ the QAOA ordering changes with depth: \aer is lower at one layer on both topologies and at two layers on the ring, while \qsim is lower at three layers on both and at two layers on the 3-regular graph. On \code{qec9xz\_n17}, the routed \qsim interval is lower than both GPU state-vector baseline intervals (Tab.~\ref{tab:app-workloads}).

\subsection{Backend costs}\label{sec:eval-resource}

On GH200, an external device trace shows comparable cuTensor-permutation and cuBLAS-GEMM time in \code{rand32x20}. On MI250X, \code{rand32x50} pages through managed memory and takes 627~s against 148~s for the matched-FP32 QuEST-GPU build; \qsim's two-buffer working set does not fit in GPU memory, whereas QuEST's in-place state does. The built-in \code{fastTranspose} host region covers $91\%$ of the profiled interval, but its stream synchronisation includes preceding asynchronous work, so the value bounds permutation time rather than attributing it to a kernel. Appendix~\ref{appendix:resource-profile} gives the memory accounting and profiler scope for both systems.

\subsection{Placement policy}\label{sec:eval-placement}

The placement comparison replays the 141-entry fused stream of \code{rand30x50} at $P\in\{2,4,8\}$ ($\Qloc=29,28,27$, respectively). With the identity initial layout at $P=8$, first-free placement requires 98 promotions, compared with 42 under all-frame farthest-next-use. A second replay uses the production initial layout and victim restriction. On this stream, the bounded production planner and its full-future counterpart both require 11, 26, and 39 promotions across $P\in\{2,4,8\}$; allowing all $\Qloc$ local positions as victims instead gives 11, 23, and 38. Thus the bounded planner matches the choices of the full-future replay within the production victim set for this stream. The remaining difference from the corollary's all-frame model comes from the restricted victim set.

\subsection{Distributed execution}\label{sec:eval-mpi}

For strong scaling, the circuit remains \code{rand30x50} while the number of ranks $P$ increases; panels (c) and (d) of Figs.~\ref{fig:ella} and~\ref{fig:setonix} show the GPU and CPU results. Weak scaling instead holds the rank-local state size fixed and adds one qubit whenever $P$ doubles; Figs.~\ref{fig:mpi-scaling} and~\ref{fig:weak256} show the small-rank and large-scale results. Both experiments include local execution, placement, transport, and MPI overhead.

\paragraph{GPU strong scaling} On \ella (Fig.~\ref{fig:ella}(c) and Tab.~\ref{tab:mpi-gpu}), \qsim has the lowest interval at $P=1$ and 2, then \aer-noCuQ takes the lead at $P=4$ and 8. The \qsim interval rises from 1.97~s on one rank to 4.32~s on eight, so additional GH200 nodes do not provide a net speedup for this fixed problem; at $P=8$, however, \aer-noCuQ is ahead by only 5\%. The \custate-backed \aer series is second at $P=1$, has no completed $P=2$ cell, leads \qsim at $P=4$, and trails it at $P=8$.

The \setonix result is qualitatively different (Fig.~\ref{fig:setonix}(c) and Tab.~\ref{tab:setonix-hip-mpi-detail}). All three engines improve from one to eight ranks, and \qsim is fastest at every rank. Its interval falls from 8.63 to 2.55~s, a $3.4\times$ strong-scaling speedup. Across the four rank counts, \aer-GPU takes 1.8--2.1 times the \qsim interval and QuEST-GPU takes 4.0--6.8 times the interval.

\paragraph{CPU strong scaling} On \ella (Fig.~\ref{fig:ella}(d) and Tab.~\ref{tab:mpi-cpu}), \qsim leads at $P=1$ and 2, but QuEST crosses below it at $P=4$. By $P=8$, QuEST takes 88.15~s against \qsim's 151.28~s. \aer is slower than \qsim at every completed rank, by factors from 3.3 to 10.4. On \setonix (Fig.~\ref{fig:setonix}(d) and Tab.~\ref{tab:setonix-mpi-cpu}), \qsim leads both competitors at every rank and improves from 245.46 to 86.83~s, a $2.8\times$ speedup. It stays 14--17\% below QuEST at $P=2$--8, while \aer takes 4.8--7.8 times the \qsim interval.

The platform split is therefore consistent across processor types: from one to eight ranks, \qsim obtains a net speedup on both \setonix paths, whereas both \ella paths end slower than their one-rank runs and cede the lead at $P=4$.

\paragraph{GPU weak scaling} The qsim-only small-rank comparison in Fig.~\ref{fig:mpi-scaling} separates the two systems: at $P=8$, the normalised interval is $2.45\times$ its one-rank value on \setonix and $11.8\times$ on \ella. The larger \setonix comparison in Fig.~\ref{fig:weak256}(a) and Tab.~\ref{tab:setonix-gpu-weak} reaches 256 ranks with $\Qloc=30$. \qsim leads \aer-GPU from $P=8$ through 32; the ordering reverses between $P=32$ and 64, and \aer-GPU remains ahead through $P=256$. Both engines reach $N=38$, corresponding to a 2~TiB FP32 state. From $P=32$ to 256, the \qsim interval grows by 52\%.

\paragraph{CPU weak scaling} In Fig.~\ref{fig:mpi-scaling}, the \qsim interval at $P=8$ is $3.0\times$ its one-rank value on \setonix and $14.0\times$ on \ella. Fig.~\ref{fig:weak256}(b) and Tab.~\ref{tab:setonix-cpu-weak} extend the \setonix comparison at $\Qloc=28$. QuEST is narrowly ahead at $P=2$, taking 38.63~s versus \qsim's 39.30~s; \qsim takes the lead at $P=4$ and retains it through the last measured QuEST cell at $P=128$. \qsim is ahead of \aer throughout, with the \aer-to-\qsim ratio narrowing from 8.5 at $P=2$ to 3.4 at $P=256$. At that final rank count, \qsim reaches $N=36$, a 1~TiB FP64 state.

\begin{figure}[t]
\centering
\includegraphics[width=\columnwidth]{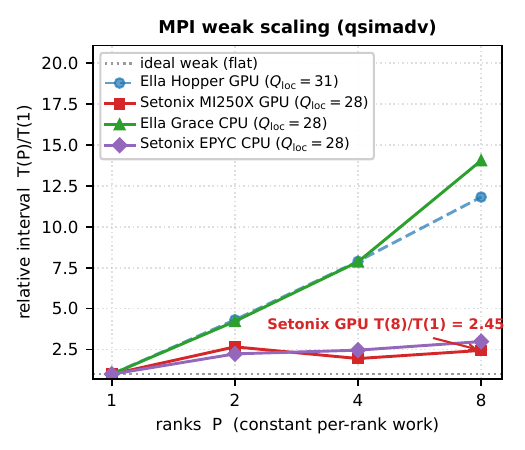}
\caption{\qsim small-rank MPI weak scaling on 50-layer random brickwork, normalised as $T(P)/T(1)$, where $T(P)$ is the interval on $P$ ranks. Solid series hold $\Qloc=28$ on \setonix GPU and both CPUs, the dashed \ella GPU series holds $\Qloc=31$, and the dotted line is ideal. The \ella Grace CPU series rises fastest, to $14\times$ at $P=8$.}
\label{fig:mpi-scaling}
\end{figure}

\begin{figure}[t]
\centering
\includegraphics[width=\columnwidth]{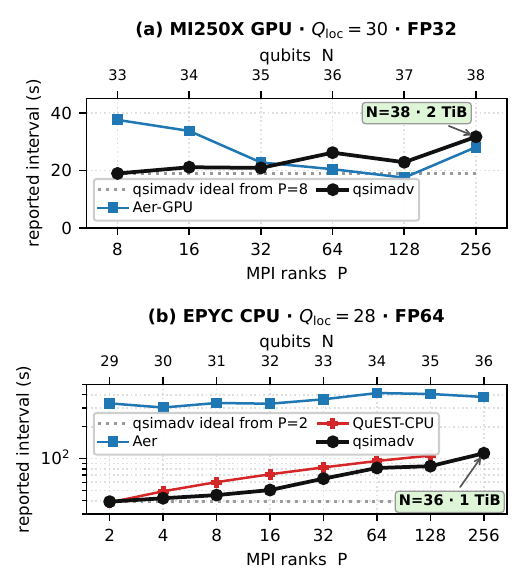}
\caption{\setonix weak scaling to 256 ranks on depth-20 random brickwork: (a) GPU at $\Qloc=30$ and FP32 reaches $N=38$ (2~TiB); (b) CPU at $\Qloc=28$ and FP64 reaches $N=36$ (1~TiB).}
\label{fig:weak256}
\end{figure}

\section{Related work}\label{sec:related}

We compare \qsim along two axes: when a simulator selects its representation and where vendor-specific code enters the state-vector path.

\subsection{NVIDIA \textup{\cuq} and its consumers}

The NVIDIA \cuq SDK provides \custate for state-vector operations and \custn for tensor-network contraction~\cite{bayraktar2023cuquantum}. \aer's GPU backend~\cite{wood2024aerprivacy}, \lightninggpu~\cite{asadi2024lightning}, \qsimg-GPU~\cite{isakov2021qsim}, and CUDA-Q~\cite{brown2026multigpu} use \custate for their fastest NVIDIA state-vector paths. The SDK delivers high throughput on NVIDIA systems but does not define an AMD backend.

Brown et al.~\cite{brown2026multigpu} study \cuq scaling across PCIe, NVLink, multi-node NVLink, and InfiniBand GPUDirect remote direct memory access (RDMA). They also discuss CUDA-to-HIP changes for inter-process communication, but report no AMD measurements. Our AMD implementation therefore starts below the quantum-library layer: ROCm supplies matrix, permutation, and communication primitives, while \qsim supplies the gate schedule and layout semantics.

\subsection{Open-source state-vector simulators}

\paragraph{\textup{\aer}} \aer is the main multi-representation baseline within Qiskit~\cite{aleksandrowicz2019qiskit}. Its default \code{automatic} method may select a stabiliser engine~\cite{wood2024aerprivacy}; an explicit state-vector/GPU method does not. \qsim extends inspection to eligible sampled-count requests made through its generic \code{dense} and \code{gpu} entry points. \aer uses Monte-Carlo trajectories for non-Clifford noise, while \qsim also includes an exact density-matrix path.

\paragraph{\textup{\lightning}} PennyLane \lightning~\cite{asadi2024lightning} supports NVIDIA and AMD GPUs through native CUDA and Kokkos-based paths. \qsim instead maps its tensor-axis interface directly to CUDA, HIP, and CPU primitives, providing a comparison between framework-level and simulator-interface portability.

\paragraph{\textup{\qsimg{}} (Google)} Google's Schr\"odinger-style \qsimg simulator~\cite{isakov2021qsim} provides native CUDA and HIP state-vector backends, a C++ MPS simulator, an optional \cuq path, and the hybrid Schr\"odinger--Feynman \code{qsimh} engine~\cite{qsimv022}. The companion qFlex work~\cite{villalonga2020establishing} provides the tensor-network precedent for \qsim's \code{cutn} engine.

\paragraph{\textup{\juqcs}} De Raedt et al. introduced a distributed simulator in 2007~\cite{deraedt2007}. \juqcs-E later used exact complex amplitudes, while \juqcs-A used an adaptive polar encoding to reach 48-qubit Shor simulations on Sunway TaihuLight~\cite{juqcs2018}. Its high-bit/low-bit state partition is a predecessor of \qsim's explicit distributed layout.

\paragraph{Multi-GPU systems} Communication-avoiding state-vector simulation spans the H\"aner--Steiger CPU implementation~\cite{haner2017half}, HyQuas~\cite{zhang2021hyquas}, UniQ~\cite{zhang2022uniq}, \nwqsim~\cite{li2023nwqsim}, and Atlas~\cite{xu2024atlas}. HyQuas includes a transpose-then-GEMM strategy, and Atlas re-stages logical-to-physical mappings between subcircuits with an integer-linear-programming partitioner. \qsim retains each resulting permutation in a per-gate layout map and implements the same contract on CUDA and HIP.

Lazy Qubit Reordering~\cite{tabuchi2024lazy} also defers cross-rank movement through time-space tiling. \qsim contributes a placement invariant and an optimal promotion count for its ideal demand model; its production policy remains a bounded-lookahead heuristic. Q$^2$Chemistry~\cite{zhong2025q2chemistry} adds buffered overlap and dependency-aware contraction on CUDA/NCCL.

Tab.~\ref{tab:feature-matrix} compares execution modes, automatic stabiliser routing, and direct or framework-mediated HIP support.

\begin{table*}[t]
\caption{Feature comparison across selected open-source quantum simulators. Single-GPU denotes a non-distributed GPU backend. \checkmark = supported; $\sim$ = supported through a portability framework or compiler vectorisation; $\circ$ = path-dependent; NV = NVIDIA-only; $\times$ = not supported.}
\label{tab:feature-matrix}
\centering\scriptsize
\begin{tabular}{l c c c c c c c c c c c}
\toprule
Simulator & MPI & Multi- & SIMD & Single- & HIP & direct & \openqasm & Stab. & Density & MPS & Tensor- \\
          &     & GPU    & CPU  & GPU     &     & HIP & native & auto-route & matrix &     & network \\
\midrule
\aer~\cite{wood2024aerprivacy}                            & \checkmark & \checkmark & \checkmark & \checkmark & \checkmark  & \checkmark & via Qiskit  & $\circ$  & \checkmark & \checkmark & NV \\
\lightning~\cite{asadi2024lightning}                      & \checkmark & \checkmark & \checkmark & \checkmark & $\sim$      & $\times$ & via PennyLane & $\times$ & $\times$  & NV        & NV \\
\qsimg~\cite{isakov2021qsim,qsimv022}                    & $\times$  & $\circ$   & \checkmark & \checkmark & \checkmark  & \checkmark & via Cirq    & $\times$ & $\times$  & \checkmark  & $\times$ \\
\juqcs~\cite{juqcs2018}                                  & \checkmark & $\times$  & \checkmark & $\times$  & $\times$    & $\times$ & via translator & $\times$ & $\times$  & $\times$  & $\times$ \\
NWQ-Sim (SV-/DM-Sim)~\cite{li2023nwqsim,li2020dmsim,nwqsimsoftware} & \checkmark & \checkmark & \checkmark & \checkmark & \checkmark  & \checkmark & \checkmark  & $\times$ & \checkmark & $\times$ & $\times$ \\
Atlas~\cite{xu2024atlas}                                  & \checkmark & \checkmark & $\times$   & \checkmark & $\times$    & $\times$ & \checkmark  & $\times$ & $\times$  & $\times$  & $\times$ \\
Q$^2$Chemistry~\cite{zhong2025q2chemistry}                & \checkmark & \checkmark & \checkmark & \checkmark & $\times$    & $\times$ & $\times$    & $\times$ & $\times$  & $\times$  & $\times$ \\
Qulacs~\cite{suzuki2021qulacs}                            & $\circ$   & $\times$  & \checkmark & \checkmark & $\times$    & $\times$ & $\circ$     & $\times$ & \checkmark & $\times$  & $\times$ \\
QuEST~\cite{jones2019quest}                               & \checkmark & \checkmark & \checkmark & \checkmark & \checkmark  & \checkmark & $\times$  & $\times$ & \checkmark & $\times$  & $\times$ \\
Intel-QS~\cite{smelyanskiy2016qhipster}                   & \checkmark & $\times$  & \checkmark & $\times$  & $\times$    & $\times$ & $\times$    & $\times$ & $\times$  & $\times$  & $\times$ \\
mpiQulacs~\cite{tabuchi2023mpiqulacs}                     & \checkmark & $\times$  & \checkmark & $\times$  & $\times$    & $\times$ & $\times$    & $\times$ & $\times$  & $\times$  & $\times$ \\
ProjectQ~\cite{steiger2018projectq}                       & $\times$  & $\times$  & \checkmark & $\times$  & $\times$    & $\times$ & $\times$    & $\times$ & $\times$  & $\times$  & $\times$ \\
MindQuantum~\cite{mindquantum_docs}                       & $\times$  & $\times$  & \checkmark & \checkmark & $\times$    & $\times$ & $\circ$     & $\times$ & \checkmark & \checkmark & $\times$ \\
Yao.jl~\cite{yao_docs}                                    & $\times$  & $\times$  & $\sim$    & \checkmark & $\times$    & $\times$ & $\times$    & $\times$ & \checkmark & $\times$  & $\times$ \\
\midrule
\textbf{\qsim (this work)}                                  & \checkmark & \checkmark & \checkmark & \checkmark & \checkmark & \checkmark & \checkmark & \checkmark & \checkmark & \checkmark & NV \\
\bottomrule
\end{tabular}

\vspace{2pt}
\begin{minipage}{\textwidth}\footnotesize
$\circ$~\aer{} Stab.: auto-selects \code{stabilizer} only under the default \code{method=automatic}; an explicit state-vector/GPU method suppresses it.\quad
$\circ$~\qsimg{} MultiGPU: only via NVIDIA \custate{} (\cuq); the native CUDA path is single-device.\quad
$\circ$~Qulacs MPI: dedicated \code{USE\_MPI} build only; power-of-two ranks, no measurement, incompatible with the GPU backend.\quad
$\circ$~Qulacs/MindQuantum \openqasm: in-tree converter (\code{qulacs.converter} / \code{mindquantum.io.OpenQASM}) for a limited gate subset, not a full parser.\\
\lightning{} HIP ($\sim$): AMD-GPU execution is provided by \lightningkokkos{} through the Kokkos HIP backend; \lightninggpu{} is NVIDIA-only.\quad
Yao SIMD ($\sim$): CPU vectorisation is supplied by the Julia/LLVM compiler rather than hand-written intrinsic kernels.\quad \\
\juqcs: the row reflects the cited CPU-only \juqcs-E/A~\cite{juqcs2018}; the later \juqcs-G line adds NVIDIA multi-GPU.
\end{minipage}
\end{table*}

\FloatBarrier

\subsection{Backend portability and native specialisation}

High-performance-computing portability frameworks such as Kokkos~\cite{kokkos2014}, OpenMP target offload~\cite{openmp51spec}, and SYCL~\cite{khronos2020sycl} compile a common kernel source for multiple targets; \lightningkokkos is a production quantum-simulation example~\cite{asadi2024lightning}. These approaches reduce the cost of supporting new devices. Our priority is per-platform performance on the supported target families. A permutation or dense update streams the full $2^N$-amplitude state, so a suboptimal lowering repeats its cost across every amplitude. We therefore do not place a common-source kernel model at the primary backend boundary.

Instead, \qsim makes the simulator interface portable and the lowering backend-specific. R2 and R3 fix the ordered-target, operator, and layout semantics, while each backend selects its native data movement, matrix library, arithmetic mode, fusion width, and transport. CUDA uses cuTensor for permutation; HIP uses a custom kernel because the evaluated hipTensor path does not support the required datatype; CPUs use their own direct kernels and NVPL or OpenBLAS. The small macro layer only maps operation names to vendor symbols at compile time. A new backend must implement the same contract, but need not inherit another backend's algorithmic choices. This accepts more backend code and validation in exchange for direct control of hardware-sensitive operations.

The CPU implementation follows the tiered instruction-set fallback pattern familiar from BLIS and OpenBLAS~\cite{vanzee2015blis}. State-vector direct kernels use NEON or AVX2 specialisations where available and otherwise rely on compiler-vectorised or scalar code. The tableau groups 64 row values in each unsigned integer and dispatches its bitwise XOR, XOR--AND, and XOR--AND-NOT operations across AVX-512, AVX2, SVE/SVE2, NEON, and scalar implementations; these operations follow Aaronson--Gottesman's analysis of tableau updates~\cite{aaronson2004improved}.

The evaluation shows why this control matters. In Fig.~\ref{fig:setonix}(a), the \lightningkokkos path is above the three FP32 paths at every shared random-brickwork width and takes 16.6 times the \qsim interval at $N=30$. Because \lightningkokkos uses FP64 and differs from \qsim above the kernel layer, the figure does not isolate Kokkos, but it does report the performance of the production portable path evaluated here. Within \qsim, the Hopper fusion sweep gives complementary evidence: changing the backend parameter $K_{\max}$ from 5 to 7 yields a $1.67\times$ speedup (\S\ref{sec:eval-3way}). Together, the system comparison and the in-system sweep motivate retaining direct control over each lowering.

\section{Discussion}\label{sec:discussion}

Semantic routing produces the largest change in the experiments because it replaces exponential state evolution with a polynomial tableau. R1 lets a sampled-count request through a generic full-state entry point reach the tableau path without an explicit representation choice. Once a circuit stays on the state-vector path, fusion, permutation, matrix arithmetic, sampling, and communication all contribute to the reported execution interval.

Both single-GPU random-brickwork sweeps exhibit a crossover: \aer is faster at smaller widths, whereas \qsim leads at larger widths, with the crossing occurring at different widths on the two systems. The measurements establish the crossover but do not isolate its component-level cause. Strong scaling likewise changes the ordering across fabrics: \qsim holds the GPU lead at every \setonix rank but cedes it from $P=4$ on \ella despite using the same logical layout scheme.

\subsection{Model and measurement limits}\label{sec:disc-gap}

Within Corollary~\ref{cor:belady}'s all-frame demand model, farthest-next-use minimises promotions for a fixed fused stream and initial layout, uniform costs, full future knowledge, and no prefetch. Production GPU-MPI instead uses bounded lookahead and a restricted victim set, while CPU-MPI uses first-free placement. The corollary is therefore a reference count, not a wall-time model; the measured MPI intervals cover the complete distributed path, including transport and local kernels, with no standalone timing ablation for the layout map.

Deterministic state comparisons under strict FP32 remain future work. For the MI250X \code{rand32x50} capacity case, an in-place single-buffer lowering would reduce the working set; separating permutation, dense updates, synchronisation, and paging would require HIP events or a device profile. The GPU-MPI ordering across the two fabrics motivates evaluating a specialised pairwise bit swap and topology-aware rank-bit assignment, neither of which is isolated by the current experiments.

\subsection{Architectural lessons}\label{sec:disc-lessons}

The architecture exposes decisions that simulators often hide inside an engine factory. Semantic routing is placed after circuit, noise, and output inspection but before state allocation, where it can change the representation without changing the request. Layout is stored as part of the simulator state, so local kernels and MPI dispatch can retain non-canonical orders without redefining logical qubits. The backend boundary is deliberately small: upper layers specify ordered targets, operators, and layouts, while native kernels, arithmetic modes, and transports remain backend choices. Together, semantic routing reaches the tableau path, while the operator and layout contracts span vendor-native full-state and distributed execution.

\section{Conclusion}\label{sec:conclusion}

\qsim treats simulator portability as controlled late binding, rather than a single implementation path imposed on every machine. Circuit, noise, and output semantics select the representation before memory is committed; ordered fused operators bind to concrete kernels only after their physical targets are known; the logical layout binds to data movement only when an operation demands non-local state. These contracts preserve circuit semantics while allowing representation, arithmetic, permutation, and transport to be selected independently where the required information becomes available.

The resulting architecture spans a polynomial tableau and exponential full-state execution, NVIDIA and AMD GPUs, Arm and x86 CPUs, and local and MPI-distributed memory without collapsing them into one lowest-common-denominator kernel. A backend implements a compact contract for matrix multiplication, permutation, and sampling, with exchange added for distributed execution. Dense arithmetic can therefore be delegated to mature, production-tested libraries such as cuBLAS, hipBLAS, NVPL, and OpenBLAS, while native kernels cover workloads that need a more specialised treatment. The single-device crossovers at larger problem sizes and the 256-rank runs show that the shared interface is not merely portable: it can retain competitive performance while exposing platform-specific capabilities.

The same separation provides a practical route to hardware beyond the systems evaluated here. A new HPC backend can reuse semantic routing, fusion, and logical layout, then bind the primitive contract to its matrix library, tensor operations, instruction set, and communication fabric. An embedded or edge backend can use the same upper layers with compact SIMD or accelerator kernels, and can avoid an exponential allocation whenever circuit semantics permit a smaller representation. Porting becomes a backend integration task rather than a rewrite of the simulator. More broadly, vendor independence need not require generic execution: it can be established by stable semantic contracts above the hardware boundary while performance-critical execution remains native below it.

\newpage
\appendices
\section{Artefact and experimental environment}\label{appendix:artifact}

The simulator and cluster recipes are hosted at \url{https://github.com/Quleaf/qsimadv}. In the simulator repository, inputs and build/run recipes are under \code{data/input/} and \code{tests/benchmark/}; the manuscript repository contains the source CSVs under \code{data/} and plotting scripts under \code{scripts/}. \code{data/README.md} maps each result family to its source files and campaign provenance and documents figure regeneration. 

\subsection*{Platforms}

\paragraph{\ella (NVIDIA Grace Hopper GH200)} Each node contains one GH200 Superchip with a 72-core Grace ARM CPU, a Hopper GPU with 96~GiB of third-generation high-bandwidth memory (HBM3), and coherent NVLink-C2C memory access. Inter-node communication uses dual ConnectX-7 interfaces over RoCEv2; the evaluated MPI path uses HPC-X OpenMPI over UCX with GPUDirect RDMA. Strong-scaling runs assign one node per rank.

\paragraph{\setonix (AMD MI250X with HPE Slingshot)} Each GPU node combines one 64-core EPYC 7A53 host with four MI250X cards, exposed as eight logical GPU compute dies (GCDs), each with 64~GiB HBM2e. Inter-node communication uses HPE Slingshot~11 and the libfabric CXI provider. The evaluated MPI path uses Cray MPICH 8.1.32; GPU-aware transfers use its HSA GPU Transfer Library (GTL).

\paragraph{\setonix CPU-only partition} The single-node CPU and CPU-MPI experiments use AMD EPYC 7763 ``Milan'' nodes from the \code{work} partition. Each node presents 128 physical cores across eight non-uniform memory access domains; the experiments assign 64 OpenMP threads per rank. These results are reported separately from those obtained on the EPYC 7A53 ``Trento'' hosts of the GPU nodes.

\subsection*{Software}

The \qsim NVIDIA build uses NVHPC SDK 26.3 (CUDA 13.1, cuBLAS, cuTensor 2.6, NVPL, and HPC-X), gcc 13.3.0, Eigen 3.4.0, and Boost; host BLAS is NVPL. The Hopper \aer baselines use CUDA 12.8 and Python 3.10.14. The AMD build uses ROCm 7.2.3, Cray MPICH, the site Eigen/OpenBLAS/Boost modules, and gcc 14.2.0. The paired \ella single-GPU and \setonix single-GCD \qsim builds use revisions \code{07fef3aa} and \code{acaeca61}, respectively; \code{data/README.md} records the full hashes. Baseline versions are \aer 0.17.2, QuEST 4.3, Qulacs 0.6.13, \lightninggpu 0.44.0, and \lightningkokkos 0.45.0.

\subsection*{Data and provenance}

Workload circuits are committed as \openqasm files. A fixed-seed generator reproduces the random-brickwork inputs \code{rand\{N\}x\{d\}}; GHZ circuits and small fixtures are committed verbatim. Raw application rows are included in the manuscript repository, and raw rows for the paired \ella single-GPU and \setonix single-GCD trials are versioned in the simulator repository. For the remaining series, the public artefact contains the aggregate CSVs and provenance records used by the paper; the authors retain the scheduler logs.

The measurements can be rerun on \ella and \setonix or comparable HPC systems, although exact timings remain system-dependent. Appendix~\ref{appendix:reproducibility} defines the timing, repetition, and numerical protocols.

\section{Measurement protocol}\label{appendix:reproducibility}

\subsection*{Timing and repetitions}

\qsim's \code{compute\_time\_s} includes gate execution, distributed exchange, and terminal sampling; it excludes process launch, parsing, fusion-plan construction, and result output. \aer times \code{.run()} after circuit parsing and simulator construction, with a warm call where noted. \lightning times QNode execution after construction. QuEST times its gate-evolution loop and excludes environment/register creation, initialisation, and sampling. Cross-tool ratios compare these internal intervals rather than a common end-to-end timer.

The application study retains its predefined runner boundaries: \qsim is measured by process wall time, including launch, parsing, fusion planning, execution, sampling, and output, whereas \aer and \lightning use the post-construction intervals above. The application table therefore compares medians under a tool-native timing policy.

Unless a caption states otherwise, each table cell is one timed observation. The paired \ella single-GPU and \setonix single-GCD sweeps, the application workloads, the \qsim GPU-MPI strong-scaling series, and the \setonix GPU small-rank curve at $\Qloc=28$ report five-run medians; the paired sweeps alternate tool order. The \qsim and QuEST 256-rank weak-scaling series report three-run medians, with each QuEST run first reduced to the maximum time across ranks. Timeouts are marked and excluded from summaries.

Most runs use $S=1000$ shots. The \setonix $\Qloc=28$ GPU weak-scaling series, the \qsim side of both 256-rank weak-scaling sweeps, and the application study use $S=100$; \aer retains $S=1000$ in the 256-rank comparisons, while QuEST's interval excludes sampling. The sampled-output checks compare distributions at their stated shot budget rather than state fidelity.

\subsection*{Numerical configuration}

The \ella and \setonix random-brickwork inputs use generator seed 0, and the per-shot random-number seed is fixed within each series. These choices reproduce the tested inputs and samples; series based on one timed observation do not estimate timing variance.

GPU state-vector rows use complex-FP32 storage for \qsim, \aer with \code{precision='single'}, and FP32 QuEST builds unless a caption states otherwise. \lightninggpu, \lightningkokkos, and Qulacs-GPU use complex-FP64 storage in the cited tables. CPU state-vector rows use complex FP64 except the explicitly marked \qsimg-CPU series. Tableau updates are bitwise and deterministic for a fixed build and input.

On GH200, \code{--gemmPrec tf32} selects TF32 multiplication with FP32 accumulation, while \code{--gemmPrec fp32} selects strict FP32 GEMM. The HIP measurements use FP32 arithmetic. State-vector sampling uses a hierarchical cumulative distribution: single-GPU execution uses $2^{10}$-entry blocks with FP64 block sums and search, and local CPU execution uses an FP64 CDF with $2^{16}$-entry blocks. MPI-GPU stores within-block prefixes in FP32 and block and inter-rank offsets in FP64; MPI-CPU uses a full FP64 local CDF.

\subsection*{Execution configuration}

Single-node CPU comparisons use 70 of 72 Grace cores on \ella and 64 threads on \setonix; each reported tool uses its fastest measured binding policy. The \setonix CPU-MPI sweep assigns one node and 64 OpenMP threads per rank.

Run records include the resolved simulation method, fusion-width cap, 64-operation DP window, cross-window merge, arithmetic mode, and distributed lowering. GPU-MPI searches 1024 entries of the gate/noise target schedule with bounded-lookahead farthest-next-use; CPU-MPI uses first-free placement. Terminal count sampling decodes the final layout map without restoring canonical state-vector order.

\FloatBarrier

\section{Supporting analysis and measurements}\label{appendix:impl}

This appendix collects the complexity summary, implementation bounds, correctness and profile evidence, and per-cell measurements underlying \S\ref{sec:eval}. Times are in seconds. Bold marks the minimum within a row when tools are columns and within a rank column when tools are rows. 

\begin{table*}[!t]
\caption{Asymptotic costs of the algorithms used in \qsim. Here $N$ is the qubit count, $k$ the fused-block width, $|C|$ the parsed instruction count, $\mathcal{W}$ the fusion-window limit, $K_{\max}$ the fusion cap, $K_{\mathrm{blk}}$ the larger of $K_{\max}$ and the width of the widest unitary noise-free input gate, $n_{\mathrm{meas}}$ the measurement count, $P=2^r$ the MPI rank count, and $\Qloc=N-r$ the local index-bit count. ``Time cost'' separates arithmetic from communication or memory traffic; ``working set'' is the dominant in-memory data structure.}
\label{tab:complexity-summary}
\centering\small
\resizebox{\textwidth}{!}{%
\begin{tabular}{l l l l}
\toprule
Algorithm                              & Time cost                                                              & Working set                  & Reference \\
\midrule
Generic permutation--GEMM (per block) & arithmetic $\Theta(2^{N+k})$ complex multiply-adds; permutation and state-update traffic each $\Theta(2^N)$ for fixed $k$                           & $\bigO{2^N+4^k}$                 & Eqs.~\eqref{eq:cost-perm}--\eqref{eq:cost-formal} \\
Diagonal fused block     & $\Theta(2^N)$                                                           & $\bigO{2^N}$                 & Eq.~\eqref{eq:cost-diag-formal} \\
Materialised density superoperator (application) & arithmetic $\bigO{2^{2N+2k}}$; state traffic $\bigO{4^N}$; operator read $\bigO{16^k}$                          & $\bigO{4^N+16^k}$                 & \S\ref{sec:methods-axisshift} \\
DP fusion (whole circuit)              & partition DP $\bigO{|C|\cdot \mathcal{W}\cdot K_{\mathrm{blk}}}$; dense block-matrix construction $\bigO{|C|\cdot 8^{K_{\max}}}$ & circuit matrices $\bigO{|C|\cdot4^{K_{\mathrm{blk}}}}$; DP auxiliary $\bigO{\mathcal{W}+K_{\mathrm{blk}}}$      & Eq.~\eqref{eq:dp-recurrence} \\
Stabiliser Clifford gate (per op)      & $\Theta(\lceil 2N/64\rceil)$ updates to 64-bit integers                                                 & $\bigO{N^2}$ bits                 & Eq.~\eqref{eq:cost-cliff-gate} \\
Stabiliser (general per shot)          & $\bigO{|C|\cdot N+(n_{\mathrm{meas}}+1)\cdot N^2}$ updates to 64-bit integers                                 & $\bigO{N^2}$ bits                 & Eq.~\eqref{eq:cost-stab-shot} \\
Lazy MPI placement (per block) & one promotion moves $\Theta(2^N/P)$ amplitudes per rank; local update $\bigO{2^{k+\Qloc}}$                & $\bigO{2^N/P+4^k+N}$    & Eq.~\eqref{eq:lazy-vol} \\
\bottomrule
\end{tabular}}
\end{table*}

\subsection{Engine bounds and concrete dispatch}\label{appendix:divergences}

The factory selects a concrete engine. State-vector and density-matrix paths instantiate their supported widths as template arguments, making dimensions, index masks, and layouts compile-time constants; MPS, tensor-network, and tableau storage is sized at runtime. Distributed state-vector paths instantiate the local width $\Qloc$ rather than the global qubit count $N$.

Tab.~\ref{tab:engine-matrix} lists the factory's dispatch ranges. Requests outside these ranges are rejected; instances within them remain subject to memory and workspace limits. The density-matrix path reuses the GPU state-vector machinery at width $2N$ under the vectorisation convention of \S\ref{sec:methods-axisshift}.

\begin{table}[!htp]
\caption{Factory dispatch ranges and state representations. The ranges are acceptance bounds, not memory guarantees.}
\label{tab:engine-matrix}
\centering\footnotesize
\begin{tabular}{@{}>{\raggedright\arraybackslash}p{0.19\columnwidth}>{\raggedright\arraybackslash}p{0.22\columnwidth}>{\raggedright\arraybackslash}p{0.49\columnwidth}@{}}
\toprule
Engine or path & Accepted range & State representation \\
\midrule
\code{dense}            & $3\le N\le39$       & CPU state vector with complex FP64 amplitudes \\
\code{gpu}              & $3\le N\le35$       & CUDA or HIP state vector with complex FP32 amplitudes \\
\code{mps}              & $N\le200$            & Runtime-sized MPS with a bond cap and singular-value cutoff \\
\code{cutn}             & $N\sim50$            & Runtime-sized NVIDIA tensor network; absent from the evaluated builds \\
\code{density\_}\allowbreak\code{matrix} & $2\le N\le17$ & Vectorised density matrix \\
tableau (routed)        & not template-bounded    & $\bigO{N^2}$-bit stabiliser tableau \\
MPI GPU                 & $3\le\Qloc\le35$    & Distributed state vector with complex FP32 amplitudes \\
MPI CPU                 & $3\le\Qloc\le39$    & Distributed state vector with complex FP64 amplitudes \\
\bottomrule
\end{tabular}
\end{table}

Each evaluated binary selects a CPU, CUDA, or HIP backend at build time; MPI builds distribute the corresponding local state-vector path. Tab.~\ref{tab:backend-matrix} summarises their native lowerings, Appendix~\ref{appendix:artifact} gives the build recipes and software stacks, and Appendix~\ref{appendix:reproducibility} defines the numerical modes.

On the evaluated CPU path, diagonal blocks stream in place. Dense blocks use direct gather--multiply--scatter kernels for $k\le3$, and for $k=4,5$ when the local state contains at least $2^{23}$ amplitudes; all other dense blocks use host permutation and complex GEMM. CPU-MPI applies the same rule within each rank.

\subsection{Fusion configuration}\label{appendix:fusion-sensitivity}

For state-vector execution, the CUDA GPU path uses $K_{\max}=7$ at $N\ge30$ and $K_{\max}=5$ below that threshold; CPU uses $K_{\max}=5$ and HIP uses $K_{\max}=6$. The marked-diagonal cap is 6 on GPU and 5 on CPU.

In the two instantaneous-quantum-polynomial cells (\code{iqp\{26,30\}d5}) and the \code{rand30x50} ablation, the dynamic-programming partition absorbs diagonal operations into mixed fused blocks; replacing its score with a diagonal-aware variant changes the Hopper intervals by less than $2\%$.

\subsection{32-bit limits in distributed exchange}\label{appendix:int-max}

Two distinct 32-bit limits affect distributed execution. The full-permutation helper stores the destination offset of every local amplitude as a signed 32-bit integer. Because a rank holds $2^{\Qloc}$ amplitudes, this routing table is valid only for $\Qloc\le30$. The restriction belongs to this helper, which is used for $\Qloc<4$ or when unchunked execution is requested; it is not a bound on the distributed engine.

The default dispatcher instead partitions the local state into chunks and realises a permutation through intra-rank swaps and XOR-partner or coalesced cross-rank transfers. It does not construct the full 32-bit destination-index table, and transfers larger than an MPI call can represent are split into pieces. The default chunk width is $B=\Qloc-1$ for $\Qloc\ge4$, so a generic fused block must satisfy $k\le B$. The large-$\Qloc$ weak-scaling runs use this path.

The second limit is the signed 32-bit \code{count} accepted by the evaluated MPI interfaces. \qsim counts typed complex elements rather than bytes and keeps every call below \code{INT\_MAX}; for example, $2^{28}$ FP64 complex amplitudes have a 4~GiB payload but a representable element count. \aer counts bytes for its swap messages, so the evaluated FP32 GPU baselines use \code{blocking\_qubits=27} and the FP64 CPU-MPI baseline uses 26 to keep each call below 2~GiB. On the \code{rand30x50} GPU strong-scaling axis, the FP32 limit is reached at $P=2$ and $P=4$ on both platforms, so those \aer cells use the blocked exchange.

\subsection{Proof of the distributed layout invariant}\label{appendix:lazy-proof}

We expand Theorem~\ref{thm:lazy-mpi-invariant}. Let $\pi_t$ map each logical qubit to its current physical index-bit position, and let $\Pi_{\pi_t}$ be the corresponding basis permutation from \S\ref{sec:methods-mpi}. If $\ket{\psi_t}$ is the canonical logical state and $\ket{\widetilde{\psi}_t}$ the stored state, the invariant is
\begin{equation*}
  \ket{\widetilde{\psi}_t}
    = \Pi_{\pi_t}\ket{\psi_t}.
\end{equation*}
Equivalently, for every logical bit string $x$, the amplitude at the physical bit string $y$ satisfying $y_{\pi_t(q)}=x_q$ equals the logical amplitude at $x$.

\begin{lemma}[Local mapped operator]
For a fused operator $G$ on ordered logical targets $V=(q_0,\ldots,q_{k-1})$, set $A=\pi_t(V)$. If all positions in $A$ are local, an order-preserving local lowering preserves the invariant with $\pi_{t+1}=\pi_t$. If the lowering retains a permutation $\alpha_A$ of the local bit positions, the invariant is preserved with $\pi_{t+1}=\hat{\alpha}_A\circ\pi_t$, where $\hat{\alpha}_A$ leaves the rank-address positions fixed.
\end{lemma}

\begin{proof}
By definition of $\Pi_{\pi_t}$, selecting the target bits at physical positions $A=\pi_t(V)$ selects the same ordered subspace index as selecting $V$ in the logical state. Let $G_W$ denote $G$ acting on the ordered targets $W$. The order-preserving case follows from
\[
G_{\pi_t(V)}\Pi_{\pi_t}=\Pi_{\pi_t}G_V.
\]
If the lowering retains $\alpha_A$, the additional basis permutation gives
\[
\Pi_{\hat{\alpha}_A}\Pi_{\pi_t}G_V\ket{\psi_t}
=\Pi_{\hat{\alpha}_A\circ\pi_t}G_V\ket{\psi_t},
\]
which proves the second case.
\end{proof}

\begin{lemma}[Distributed promotion swap]
Suppose a target lies at rank-address position $b\ge\Qloc$ and $\ell<\Qloc$ is an unpinned local position. Let $\tau$ exchange physical positions $b$ and $\ell$. A distributed bit transpose implementing $\Pi_\tau$, followed by the map update $\pi_{t+1}=\tau\circ\pi_t$, preserves the invariant and makes that target local.
\end{lemma}

\begin{proof}
The transpose changes the stored state to
\[
\Pi_\tau\ket{\widetilde{\psi}_t}
=\Pi_\tau\Pi_{\pi_t}\ket{\psi_t}
=\Pi_{\tau\circ\pi_t}\ket{\psi_t}.
\]
The logical state is unchanged during promotion, and the updated map records exactly the physical transposition. Because $b$ is exchanged with a local position, the promoted target becomes local. The condition $k\le\Qloc$ guarantees that enough unpinned local positions exist for all targets of the fused operator.
\end{proof}

These lemmas prove the theorem by induction over promotion swaps and fused operators. A terminal computational-basis sampler need not invert the permutation: it samples a physical basis index and decodes each logical output bit through $\pi_t$. A consumer that requires canonical amplitude order may instead request materialisation of the inverse permutation.

\subsection{Precision comparison}

The first two columns of Tab.~\ref{tab:precision-bridge} hold the \aer-noCuQ backend fixed and change its configured precision from FP64 to FP32. The FP32 configuration reduces its interval by $1.46$--$1.86\times$ across the sweep. It is therefore the relevant matched-storage \aer-noCuQ baseline for \qsim: \aer is faster at $N=24$ and 26, and \qsim is faster from $N=28$.

\begin{table}[!htp]
\caption{Precision comparison on \ella; entries are five-run medians in seconds. \aer denotes \aer-noCuQ, and the ratio uses its FP32 column. \qsim stores FP32 states and uses TF32 multiplication with FP32 accumulation.}
\label{tab:precision-bridge}
\centering\footnotesize
\setlength{\tabcolsep}{3pt}
\begin{tabular}{@{}lrrrr@{}}
\toprule
circuit & \aer FP64 & \aer FP32 & \qsim FP32 & ratio \\
\midrule
\code{rand24x50} & 0.127 & 0.087 & 0.256 & 0.34$\times$ \\
\code{rand26x50} & 0.413 & 0.238 & 0.410 & 0.58$\times$ \\
\code{rand28x50} & 1.701 & 0.949 & 0.901 & 1.05$\times$ \\
\code{rand30x50} & 7.124 & 3.821 & 1.932 & 1.98$\times$ \\
\code{rand32x50} & 30.13 & 16.27 & 7.415 & 2.19$\times$ \\
\bottomrule
\end{tabular}
\end{table}

\subsection{Capacity and execution profiles}\label{appendix:resource-profile}

\paragraph{GH200} A complex FP32 state vector occupies $8\cdot2^N$ bytes, and the permutation--GEMM path holds both input and output vectors. Its two-vector baseline is therefore 16~GiB at $N=30$ and 64~GiB at $N=32$. The measured footprints are 22.5 and 70.5~GiB, respectively; the additional 6.5~GiB is library workspace. The $N=32$ case remains within the GH200's 96~GiB device memory. An external device-kernel trace of \code{rand32x20} attributes 1.95~s to cuTensor permutation and 1.81~s to cuBLAS GEMM, with negligible sampling time, showing that the two full-state terms have comparable cost in this run.

\paragraph{MI250X} The same $N=32$ two-vector baseline fills the 64~GiB available to one GCD before workspace is added, so \code{rand32x50} runs through managed-memory paging. The matched-FP32 QuEST-GPU build uses one in-place state vector and completes in 148~s, compared with 626.6~s for \qsim. The available host profile does not support a kernel-level split: \code{fastTranspose} accounts for $91\%$ of the interval because it synchronises the stream and therefore absorbs preceding asynchronous work and paging, while \code{fastMatmulGPU} records only $0.010$~s. At $N\ge32$, \qsim uses a custom dense-update kernel with 64-bit offsets because the tested rocBLAS \code{cgemm} interface uses 32-bit indices.

\subsection{Per-platform measurement tables}

The tables below provide the CPU, MPI, and large-scale weak-scaling values behind the evaluation figures. Fig.~\ref{fig:qft} combines the QFT rows from the four single-device tables.

\begin{figure*}[t]
\centering
\includegraphics[width=\textwidth]{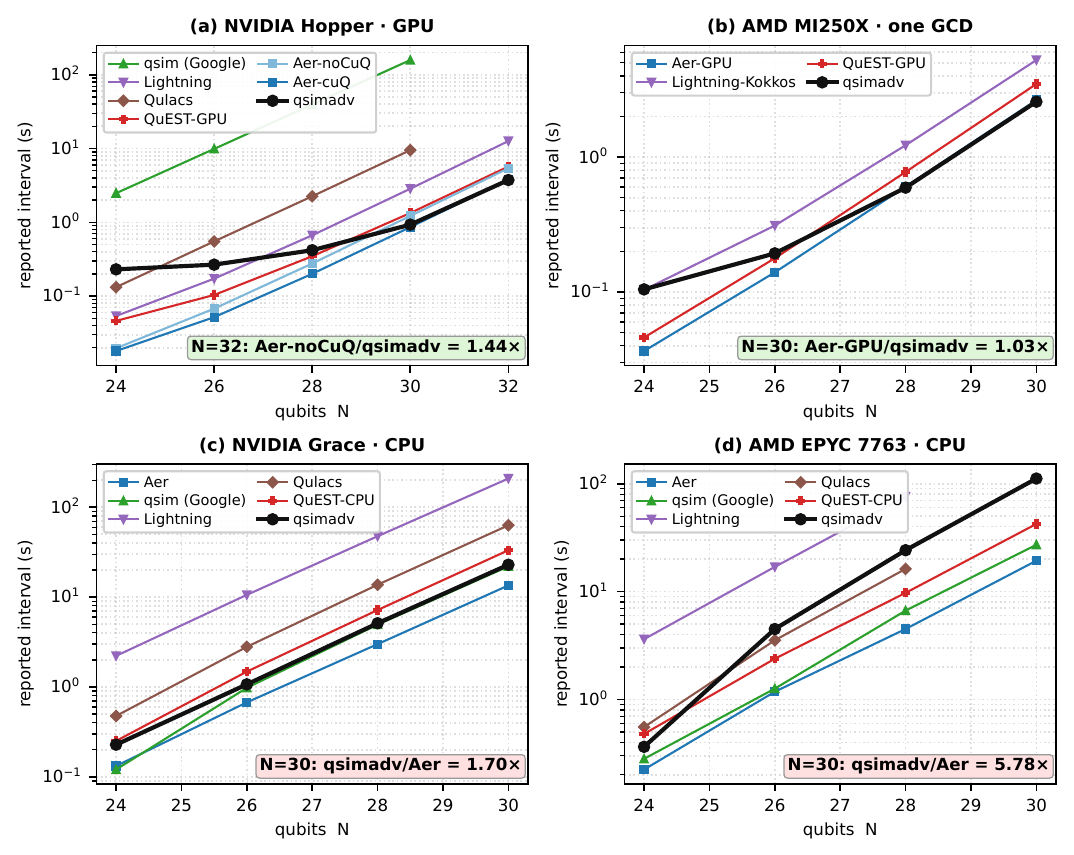}
\caption{Single-device QFT intervals on (a) \ella Hopper, (b) one \setonix MI250X GCD, (c) \ella Grace, and (d) one \setonix EPYC node. \lightninggpu, Qulacs-GPU, and \lightningkokkos use FP64; \qsimg-CPU uses FP32. Values: Tabs.~\ref{tab:3way-gpu}, \ref{tab:setonix-engines}, \ref{tab:3way-cpu}, and \ref{tab:setonix-cpu-3way}.}
\label{fig:qft}
\end{figure*}

\subsubsection{\ella}

The following tables give the \ella CPU and MPI data behind Fig.~\ref{fig:ella}; the single-GPU cells are in Tab.~\ref{tab:3way-gpu}.

\begin{table}[!htp]
\caption{Single-process CPU intervals on \ella. \qsimg-CPU uses FP32; the other engines use FP64. The panels share rows; ``T/O'' denotes timeout.}
\label{tab:3way-cpu}
\centering\footnotesize
\begin{tabular}{@{}lrrr@{}}
\toprule
Circuit       & \qsim & \aer & \qsimg-CPU \\
\midrule
\code{rand24x50}  & 1.139 & 1.044 & \good{0.531} \\
\code{rand26x50}  & 7.062 & 4.799 & \good{2.531} \\
\code{rand28x50}  & 20.85 & 22.64 & \good{11.85} \\
\code{rand30x50}  & 88.31 & 98.33 & \good{50.64} \\
\code{qft24}      & 0.228 & 0.132 & \good{0.121} \\
\code{qft26}      & 1.070 & \good{0.674} & 0.978 \\
\code{qft28}      & 5.113 & \good{2.996} & 4.938 \\
\code{qft30}      & 22.90 & \good{13.47} & 21.96 \\
\bottomrule
\end{tabular}
\vspace{3pt}

\begin{tabular}{@{}lrrr@{}}
\toprule
Circuit       & \lightning & Qulacs & QuEST-CPU \\
\midrule
\code{rand24x50}  & 80.24 & 3.133 & 3.286 \\
\code{rand26x50}  & 240.8 & 12.43 & 11.97 \\
\code{rand28x50}  & T/O & 80.43 & 78.42 \\
\code{rand30x50}  & T/O & 337.2 & 332.1 \\
\code{qft24}      & 2.217 & 0.476 & 0.251 \\
\code{qft26}      & 10.59 & 2.795 & 1.482 \\
\code{qft28}      & 47.55 & 13.72 & 7.177 \\
\code{qft30}      & 208.4 & 62.91 & 33.20 \\
\bottomrule
\end{tabular}
\end{table}

The MPI tables use HPC-X OpenMPI over RoCEv2. The GPU sweep assigns one GH200 to each rank and uses complex-FP32 storage; because \code{rand30x50} fits on one GH200, additional ranks expose distributed-execution cost rather than capacity. The \aer rows at $P\ge2$ use \code{blocking\_qubits=27} to keep swap messages below the count-width limit of Appendix~\ref{appendix:int-max}. The CPU sweep assigns one Grace node to each rank and uses FP64.

\begin{table}[!htp]
\caption{GPU-MPI strong scaling on \code{rand30x50} over \ella. \qsim cells are five-run medians; each \aer cell is one warm observation. ``T/O'' denotes timeout.}
\label{tab:mpi-gpu}
\centering\footnotesize
\begin{tabular}{@{}r r r r@{}}
\toprule
$P$ & \qsim & \aer-noCuQ & \aer-cuQ$^{\dagger}$ \\
\midrule
1 & \good{1.97} & 3.94 & 3.67 \\
2 & \good{4.28} & 4.97 & T/O \\
4 & 5.37        & \good{3.83} & 4.25 \\
8 & 4.32        & \good{4.11} & 4.66 \\
\bottomrule
\end{tabular}
\\[2pt]{\footnotesize $^{\dagger}$ closed-source \cuq.}
\end{table}

The small-rank GPU weak-scaling series in Fig.~\ref{fig:mpi-scaling} holds $\Qloc=31$ and reaches $N=34$. The \aer CPU-MPI build uses its chunked \code{blocking\_enable} configuration, so its $P=1$ row is distinct from the single-process configuration in Tab.~\ref{tab:3way-cpu}.

\begin{table}[!htp]
\caption{CPU-MPI state-vector strong scaling on \code{rand30x50} over \ella. ``T/O'' denotes timeout.}
\label{tab:mpi-cpu}
\centering\small
\begin{tabular}{@{}r r r r@{}}
\toprule
$P$ & \qsim-CPU & \aer-CPU & QuEST-CPU \\
\midrule
1 & \good{87.79}  & T/O    & 330.70 \\
2 & \good{191.68} & 2001.9 & 208.65 \\
4 & 162.00        & 975.3  & \good{133.32} \\
8 & 151.28        & 503.9  & \good{88.15} \\
\bottomrule
\end{tabular}
\end{table}

\subsubsection{\setonix}

The following tables give the \setonix CPU, strong-scaling, and capacity results behind Fig.~\ref{fig:setonix} and Fig.~\ref{fig:weak256}; the single-GCD cells are in Tab.~\ref{tab:setonix-engines}.

\begin{table}[!htp]
\caption{Single-process CPU intervals on \setonix. \qsimg-CPU uses FP32; the other engines use FP64. The panels share rows; ``T/O'' denotes timeout and ``N/M'' an unmeasured cell.}
\label{tab:setonix-cpu-3way}
\centering\footnotesize
\begin{tabular}{@{}lrrr@{}}
\toprule
Circuit       & \qsim & \aer & \qsimg-CPU \\
\midrule
\code{rand24x50}  & 2.150 & 3.071 & \good{1.164} \\
\code{rand26x50}  & 18.21 & 8.986 & \good{2.361} \\
\code{rand28x50}  & 56.98 & 40.37 & \good{19.68} \\
\code{rand30x50}  & 243.6 & 176.4 & \good{55.65} \\
\code{qft24}      & 0.364 & \good{0.224} & 0.282 \\
\code{qft26}      & 4.494 & \good{1.179} & 1.260 \\
\code{qft28}      & 24.17 & \good{4.498} & 6.658 \\
\code{qft30}      & 111.6 & \good{19.31} & 27.16 \\
\bottomrule
\end{tabular}
\vspace{3pt}

\begin{tabular}{@{}lrrr@{}}
\toprule
Circuit       & \lightning & Qulacs & QuEST-CPU \\
\midrule
\code{rand24x50}  & 71.27 & 3.878 & 3.275 \\
\code{rand26x50}  & T/O & 16.73 & 27.28 \\
\code{rand28x50}  & T/O & 92.36 & 102.1 \\
\code{rand30x50}  & T/O & T/O & 405.6 \\
\code{qft24}      & 3.620 & 0.555 & 0.479 \\
\code{qft26}      & 16.94 & 3.539 & 2.385 \\
\code{qft28}      & 76.24 & 16.22 & 9.727 \\
\code{qft30}      & N/M & N/M & 42.34 \\
\bottomrule
\end{tabular}
\end{table}

The MPI comparisons use Cray MPICH over Slingshot/CXI. GPU runs use one rank per node and one GCD per rank with closest-GPU binding; the two- and four-rank \aer cells use \code{blocking\_qubits=27} to stay below the count-width limit of Appendix~\ref{appendix:int-max}.

\begin{table}[!htp]
\caption{GPU-MPI strong scaling on \code{rand30x50} over \setonix. $^\ddagger$ marks \aer cells with reduced \code{blocking\_qubits} to keep messages representable (Appendix~\ref{appendix:int-max}).}
\label{tab:setonix-hip-mpi-detail}
\centering\footnotesize
\resizebox{\columnwidth}{!}{%
\begin{tabular}{lrrrr}
\toprule
Tool & $P=1$ [s] & $P=2$ [s] & $P=4$ [s] & $P=8$ [s] \\
\midrule
\qsim & \good{8.63} & \good{7.08} & \good{4.14} & \good{2.55} \\
\aer-GPU   & 18.35 & 12.68$^\ddagger$ & 7.48$^\ddagger$ & 4.65 \\
QuEST-GPU & 34.47 & 32.16 & 25.80 & 17.44 \\
\bottomrule
\end{tabular}
}
\end{table}

The CPU-MPI sweep assigns one EPYC node and 64 OpenMP threads to each rank; all three tools use the same Cray MPICH stack. The \aer CPU-MPI rows use chunked blocking with \code{blocking\_qubits=26}, whereas the GPU-MPI rows above use 27.

\begin{table}[!htp]
\caption{CPU-MPI state-vector strong scaling on \code{rand30x50} over \setonix.}
\label{tab:setonix-mpi-cpu}
\centering\small
\begin{tabular}{@{}r r r r@{}}
\toprule
$P$ & \qsim-CPU & \aer-CPU & QuEST-CPU \\
\midrule
1 & \good{245.46} & 3044.9 & 451.43 \\
2 & \good{232.04} & 1807.1 & 271.45 \\
4 & \good{143.57} & 872.2  & 167.97 \\
8 & \good{86.83}  & 419.7  & 104.64 \\
\bottomrule
\end{tabular}
\end{table}

For large-scale weak scaling, $N$ increases with $P$ while the rank-local state size remains fixed. The GPU sweep uses $\Qloc=30$ (one MI250X GCD per rank, complex-FP32 storage), and the CPU sweep uses $\Qloc=28$ (one EPYC node per rank, FP64).

\begin{table}[!htp]
\caption{\setonix GPU-MPI weak scaling on depth-20 random brickwork at $\Qloc=30$.}
\label{tab:setonix-gpu-weak}
\centering\small
\begin{tabular}{rrrrr}
\toprule
$P$ & $N$ & nodes & \qsim [s] & \aer-GPU [s] \\
\midrule
8   & 33 & 1  & \good{18.93} & 37.54 \\
16  & 34 & 2  & \good{21.09} & 33.70 \\
32  & 35 & 4  & \good{20.85} & 22.75 \\
64  & 36 & 8  & 26.15  & \good{20.41} \\
128 & 37 & 16 & 22.82  & \good{17.46} \\
256 & 38 & 32 & 31.66  & \good{28.07} \\
\bottomrule
\end{tabular}
\end{table}

\begin{table}[!htp]
\caption{\setonix CPU-MPI weak scaling on depth-20 random brickwork at $\Qloc=28$. ``N/M'' denotes an unmeasured cell.}
\label{tab:setonix-cpu-weak}
\centering\small
\resizebox{\columnwidth}{!}{%
\begin{tabular}{rrrrrr}
\toprule
$P$ & $N$ & total state & \qsim [s] & \aer-CPU [s] & QuEST [s] \\
\midrule
2   & 29 & 8~GiB   &  39.30 & 332.45 &  \good{38.63} \\
4   & 30 & 16~GiB  &  \good{42.32} & 303.83 &  49.21 \\
8   & 31 & 32~GiB  &  \good{45.30} & 335.26 &  60.11 \\
16  & 32 & 64~GiB  &  \good{50.58} & 329.91 &  71.23 \\
32  & 33 & 128~GiB &  \good{64.77} & 363.28 &  82.68 \\
64  & 34 & 256~GiB &  \good{81.64} & 415.89 &  95.06 \\
128 & 35 & 512~GiB &  \good{84.97} & 407.68 & 107.05 \\
256 & 36 & 1~TiB   & \good{112.44} & 381.91 & N/M \\
\bottomrule
\end{tabular}}
\end{table}

\FloatBarrier
\subsection{Semantic routing and application workloads}\label{appendix:qec-detail}

The routing benchmark contains eight GHZ circuits with per-gate depolarising noise at $p_{\mathrm{dep}}=10^{-2}$ and two noiseless random-Clifford circuits. Each \qsim input is run once with ordinary semantic routing and once with routing disabled to force state-vector execution where it fits. Every ordinary request selects the tableau; the forced state-vector series ends at $N=30$. \aer's stabiliser sampler and Stim provide dedicated stabiliser baselines in Tab.~\ref{tab:qec-stab}.

For each GHZ cell, let $\hat p_i$ be tool $i$'s observed frequency of the all-zero outcome $0^N$ at $S=1000$. For tools $i,j$, the independent binomial standard error is $\mathrm{se}_{ij}=\sqrt{\hat p_i(1-\hat p_i)/S+\hat p_j(1-\hat p_j)/S}$. The largest defined ratio $|\hat p_i-\hat p_j|/\mathrm{se}_{ij}$ is 1.02. This check covers one sampled GHZ marginal; it does not cover the random-Clifford outputs.

\begin{table}[!htp]
\caption{Synthetic Clifford routing on \ella. \qsim~tab and \qsim~SV denote routed-tableau and forced-state-vector modes; ``N/A'' denotes an unavailable state-vector cell.}
\label{tab:qec-stab}
\centering\footnotesize
\resizebox{\columnwidth}{!}{%
\begin{tabular}{@{}lrrrrr@{}}
\toprule
Circuit & $N$ & \qsim tab & \aer stab & Stim & \qsim SV \\
\midrule
\code{ghz20}        & 20   & 0.0061 & 0.0020 & \good{0.00005} & 1.284 \\
\code{ghz25}        & 25   & 0.0056 & 0.0028 & \good{0.00006} & 12.45 \\
\code{ghz30}        & 30   & 0.0058 & 0.0026 & \good{0.00007} & 454.8 \\
\code{ghz60}        & 60   & 0.0073 & 0.0045 & \good{0.00011} & N/A \\
\code{ghz100}       & 100  & 0.0097 & 0.0082 & \good{0.00020} & N/A \\
\code{ghz200}       & 200  & 0.0198 & 0.0243 & \good{0.00037} & N/A \\
\code{ghz500}       & 500  & 0.129  & 0.174  & \good{0.00090} & N/A \\
\code{ghz1000}      & 1000 & 0.884  & 1.195  & \good{0.00164} & N/A \\
\code{randcliff100} & 100  & 0.113  & 0.204  & \good{0.00041} & N/A \\
\code{randcliff500} & 500  & 4.715  & 12.81  & \good{0.00585} & N/A \\
\bottomrule
\end{tabular}}
\end{table}

The application set comprises QAOA size and depth sweeps, one-layer UCCSD-style H$_2$, LiH, and BeH$_2$ circuits, and four QASMBench inputs, including the Clifford \code{qec9xz\_n17} case. Each shared \openqasm input supported by all three tools is run with \qsim, \aer, and \lightning on one Hopper GPU. Tab.~\ref{tab:app-workloads} reports five-trial medians under the runner boundaries of Appendix~\ref{appendix:reproducibility}.

\begin{table}[!htp]
\centering\footnotesize
\caption{Five-trial application-workload medians on \ella under the tool-native timing policy of Appendix~\ref{appendix:reproducibility}. $N$ is the number of acted-on qubits.}
\label{tab:app-workloads}
\setlength{\tabcolsep}{3pt}
\begin{tabular}{@{}l c r r r@{}}
\toprule
Circuit & $N$ & \qsim & \aer & \lightning \\
\midrule
\multicolumn{5}{l}{\emph{QAOA, 2 layers: size sweep}} \\
\code{qaoa18p2\_ring} & 18 & 0.572 & 0.375 & \good{0.127} \\
\code{qaoa22p2\_ring} & 22 & 0.574 & 0.385 & \good{0.143} \\
\code{qaoa26p2\_ring} & 26 & 0.625 & 0.426 & \good{0.233} \\
\code{qaoa18p2\_3reg} & 18 & 0.583 & 0.333 & \good{0.130} \\
\code{qaoa22p2\_3reg} & 22 & 0.585 & 0.395 & \good{0.147} \\
\code{qaoa26p2\_3reg} & 26 & 0.616 & 0.424 & \good{0.265} \\
\multicolumn{5}{l}{\emph{QAOA at $N=30$: one to three layers}} \\
\code{qaoa30p1\_ring} & 30 & 0.933 & \good{0.698} & 1.094 \\
\code{qaoa30p2\_ring} & 30 & 1.013 & \good{0.988} & 1.972 \\
\code{qaoa30p3\_ring} & 30 & \good{1.085} & 1.205 & 2.849 \\
\code{qaoa30p1\_3reg} & 30 & 0.978 & \good{0.758} & 1.365 \\
\code{qaoa30p2\_3reg} & 30 & \good{1.070} & 1.107 & 2.511 \\
\code{qaoa30p3\_3reg} & 30 & \good{1.233} & 1.424 & 3.667 \\
\multicolumn{5}{l}{\emph{UCCSD-style, 1 layer}} \\
\code{uccsd\_h2} & 8 & 0.603 & 0.401 & \good{0.190} \\
\code{uccsd\_lih} & 12 & 0.756 & \good{0.658} & 1.137 \\
\code{uccsd\_beh2} & 14 & 1.279 & \good{1.048} & 2.875 \\
\multicolumn{5}{l}{\emph{QASMBench medium}} \\
\code{multiply\_n13} & 13 & 0.586 & 0.322 & \good{0.127} \\
\code{multiplier\_n15} & 15 & 0.580 & 0.327 & \good{0.130} \\
\code{dnn\_n16} & 16 & 0.591 & 0.412 & \good{0.306} \\
\code{qec9xz\_n17} & 17 & \good{0.071} & 0.430 & 0.125 \\
\bottomrule
\end{tabular}
\end{table}

\FloatBarrier
\section*{Acknowledgment}
This work was supported by resources provided by the Pawsey Supercomputing Research Centre with funding from the Australian Government and the Government of Western Australia. It was carried out within the Pawsey Supercomputing Research Centre’s Quantum Supercomputing Innovation Hub, made possible by a grant from the Australian Government through the National Collaborative Research Infrastructure Strategy (NCRIS). Computational resources were provided by the Pawsey Supercomputing Research Centre’s Setonix Supercomputer (\href{https://doi.org/10.48569/18sb-8s43}{https://doi.org/10.48569/18sb-8s43}), with funding from the Australian Government and the Government of Western Australia.

AI-assisted tools were used during the development of portions of the software implementation. All generated code and related documents were reviewed, validated, and modified by the authors.

\bibliographystyle{IEEEtran}
\IEEEtriggercmd{\newpage}
\IEEEtriggeratref{36}
\bibliography{references}

\begin{thebibliography}{10}
\providecommand{\url}[1]{#1}
\csname url@samestyle\endcsname
\providecommand{\newblock}{\relax}
\providecommand{\bibinfo}[2]{#2}
\providecommand{\BIBentrySTDinterwordspacing}{\spaceskip=0pt\relax}
\providecommand{\BIBentryALTinterwordstretchfactor}{4}
\providecommand{\BIBentryALTinterwordspacing}{\spaceskip=\fontdimen2\font plus
\BIBentryALTinterwordstretchfactor\fontdimen3\font minus
  \fontdimen4\font\relax}
\providecommand{\BIBforeignlanguage}[2]{{%
\expandafter\ifx\csname l@#1\endcsname\relax
\typeout{** WARNING: IEEEtran.bst: No hyphenation pattern has been}%
\typeout{** loaded for the language `#1'. Using the pattern for}%
\typeout{** the default language instead.}%
\else
\language=\csname l@#1\endcsname
\fi
#2}}
\providecommand{\BIBdecl}{\relax}
\BIBdecl

\bibitem{bayraktar2023cuquantum}
H.~Bayraktar, A.~Charara, D.~Clark, S.~Cohen, T.~Costa, Y.-L.~L. Fang, Y.~Gao,
  J.~Guan, J.~Gunnels, A.~Haidar, A.~Hehn, M.~Hohnerbach, M.~Jones, T.~Lubowe,
  D.~Lyakh, S.~Morino, P.~Springer, S.~Stanwyck, I.~Terentyev, S.~Varadhan,
  J.~Wong, and T.~Yamaguchi, ``{cuQuantum SDK}: a high-performance library for
  accelerating quantum science,'' in \emph{2023 IEEE International Conference
  on Quantum Computing and Engineering (QCE)}, 2023, pp. 1050--1061.

\bibitem{aleksandrowicz2019qiskit}
\BIBentryALTinterwordspacing
G.~Aleksandrowicz \emph{et~al.}, ``Qiskit: An open-source framework for quantum
  computing,'' Zenodo software record, 2019. [Online]. Available:
  \url{https://doi.org/10.5281/zenodo.2562111}
\BIBentrySTDinterwordspacing

\bibitem{wood2024aerprivacy}
\BIBentryALTinterwordspacing
{Qiskit Development Team}, ``{AerSimulator}: simulation method options,'' 2025,
  qiskit Aer documentation, accessed: Jul. 31, 2026. [Online]. Available:
  \url{https://qiskit.github.io/qiskit-aer/stubs/qiskit_aer.AerSimulator.html}
\BIBentrySTDinterwordspacing

\bibitem{asadi2024lightning}
\BIBentryALTinterwordspacing
A.~Asadi, A.~Dusko, C.-Y. Park, V.~Michaud-Rioux, I.~Schoch, S.~Shu,
  T.~Vincent, and L.~J. O'Riordan, ``Hybrid quantum programming with {PennyLane
  Lightning} on {HPC} platforms,'' arXiv preprint arXiv:2403.02512, 2024.
  [Online]. Available: \url{https://arxiv.org/abs/2403.02512}
\BIBentrySTDinterwordspacing

\bibitem{isakov2021qsim}
\BIBentryALTinterwordspacing
S.~V. Isakov, D.~Kafri, O.~Martin, C.~{Vollgraff Heidweiller}, W.~Mruczkiewicz,
  M.~P. Harrigan, N.~C. Rubin, R.~Thomson, M.~Broughton, K.~Kissell, E.~Peters,
  E.~Gustafson, A.~C.~Y. Li, H.~Lamm, G.~Perdue, A.~K. Ho, D.~Strain, and
  S.~Boixo, ``Simulations of quantum circuits with approximate noise using
  {qsim} and {Cirq},'' arXiv preprint arXiv:2111.02396, 2021. [Online].
  Available: \url{https://arxiv.org/abs/2111.02396}
\BIBentrySTDinterwordspacing

\bibitem{brown2026multigpu}
W.~M. Brown, A.~Ramesh, T.~Lubinski, T.~Nguyen, and D.~E. Bernal~Neira,
  ``Multi-gpu quantum circuit simulation and the impact of network
  performance,'' \emph{Computer Physics Communications}, vol. 324, p. 110126,
  2026.

\bibitem{setonix}
\BIBentryALTinterwordspacing
{Pawsey Supercomputing Research Centre}, ``{Setonix},'' n.d., accessed: Jul.
  14, 2026. [Online]. Available: \url{https://pawsey.org.au/systems/setonix/}
\BIBentrySTDinterwordspacing

\bibitem{rocm721release}
\BIBentryALTinterwordspacing
{Advanced Micro Devices}, ``{ROCm 7.2.3} release notes,'' 2026. [Online].
  Available:
  \url{https://rocm.docs.amd.com/en/docs-7.2.3/about/release-notes.html}
\BIBentrySTDinterwordspacing

\bibitem{vidal2003efficient}
G.~Vidal, ``Efficient classical simulation of slightly entangled quantum
  computations,'' \emph{Phys. Rev. Lett.}, vol.~91, no.~14, p. 147902, 2003.

\bibitem{gottesman1998heisenberg}
D.~Gottesman, ``The heisenberg representation of quantum computers,'' in
  \emph{Group22: Proceedings of the XXII International Colloquium on Group
  Theoretical Methods in Physics}.\hskip 1em plus 0.5em minus 0.4em\relax
  International Press, 1999, pp. 32--43.

\bibitem{aaronson2004improved}
S.~Aaronson and D.~Gottesman, ``Improved simulation of stabilizer circuits,''
  \emph{Phys. Rev. A}, vol.~70, no.~5, p. 052328, 2004.

\bibitem{suzuki2021qulacs}
Y.~Suzuki, Y.~Kawase, Y.~Masumura, Y.~Hiraga, M.~Nakadai, J.~Chen, K.~M.
  Nakanishi, K.~Mitarai, R.~Imai, S.~Tamiya, T.~Yamamoto, T.~Yan, T.~Kawakubo,
  Y.~O. Nakagawa, Y.~Ibe, Y.~Zhang, H.~Yamashita, H.~Yoshimura, A.~Hayashi, and
  K.~Fujii, ``{Qulacs}: a fast and versatile quantum circuit simulator for
  research purpose,'' \emph{Quantum}, vol.~5, p. 559, 2021.

\bibitem{jones2019quest}
T.~Jones, A.~Brown, I.~Bush, and S.~C. Benjamin, ``{QuEST} and high performance
  simulation of quantum computers,'' \emph{Scientific Reports}, vol.~9, no.~1,
  p. 10736, 2019.

\bibitem{juqcs2018}
H.~De~Raedt, F.~Jin, D.~Willsch, M.~Willsch, N.~Yoshioka, N.~Ito, S.~Yuan, and
  K.~Michielsen, ``Massively parallel quantum computer simulator, eleven years
  later,'' \emph{Computer Physics Communications}, vol. 237, pp. 47--61, 2019.

\bibitem{li2023nwqsim}
A.~Li, B.~Fang, C.~Granade, G.~Prawiroatmodjo, B.~Heim, M.~Roetteler, and
  S.~Krishnamoorthy, ``{SV-Sim}: Scalable {PGAS}-based state vector simulation
  of quantum circuits,'' in \emph{Proceedings of the International Conference
  for High Performance Computing, Networking, Storage and Analysis}, 2021, pp.
  1--14.

\bibitem{xu2024atlas}
M.~Xu, S.~Cao, X.~Miao, U.~A. Acar, and Z.~Jia, ``{Atlas}: Hierarchical
  partitioning for quantum circuit simulation on {GPUs},'' in \emph{SC24:
  International Conference for High Performance Computing, Networking, Storage
  and Analysis}, 2024, pp. 1--17.

\bibitem{gidney2021stim}
C.~Gidney, ``{Stim}: a fast stabilizer circuit simulator,'' \emph{Quantum},
  vol.~5, p. 497, 2021.

\bibitem{quimb2018}
J.~Gray, ``{quimb}: a {Python} package for quantum information and many-body
  calculations,'' \emph{Journal of Open Source Software}, vol.~3, no.~29, p.
  819, 2018.

\bibitem{tensorcircuit2023}
S.-X. Zhang, J.~Allcock, Z.-Q. Wan, S.~Liu, J.~Sun, H.~Yu, X.-H. Yang, J.~Qiu,
  Z.~Ye, Y.-Q. Chen, C.-K. Lee, Y.-C. Zheng, S.-K. Jian, H.~Yao, C.-Y. Hsieh,
  and S.~Zhang, ``{TensorCircuit}: a quantum software framework for the {NISQ}
  era,'' \emph{Quantum}, vol.~7, p. 912, 2023.

\bibitem{qsimv022}
\BIBentryALTinterwordspacing
{Google Quantum AI}, ``{qsim} v0.22.0 documentation: {AMD GPU} and {MPS}
  support,'' Software documentation, 2026, accessed 2026-07-15. [Online].
  Available: \url{https://github.com/quantumlib/qsim/tree/v0.22.0/docs}
\BIBentrySTDinterwordspacing

\bibitem{cross2022openqasm}
A.~W. Cross, A.~Javadi-Abhari, T.~Alexander, N.~de~Beaudrap, L.~S. Bishop,
  S.~Heidel, C.~A. Ryan, P.~Sivarajah, J.~Smolin, J.~M. Gambetta, and B.~R.
  Johnson, ``{OpenQASM 3}: A broader and deeper quantum assembly language,''
  \emph{ACM Transactions on Quantum Computing}, vol.~3, no.~3, pp. 1--50, 2022.

\bibitem{springer2018tccg}
P.~Springer and P.~Bientinesi, ``Design of a high-performance {GEMM}-like
  tensor--tensor multiplication,'' \emph{ACM Transactions on Mathematical
  Software}, vol.~44, no.~3, pp. 1--29, 2018.

\bibitem{belady1966study}
L.~A. Bel{\'a}dy, ``A study of replacement algorithms for a virtual-storage
  computer,'' \emph{IBM Systems Journal}, vol.~5, no.~2, pp. 78--101, 1966.

\bibitem{mattson1970evaluation}
R.~L. Mattson, J.~Gecsei, D.~R. Slutz, and I.~L. Traiger, ``Evaluation
  techniques for storage hierarchies,'' \emph{IBM Systems Journal}, vol.~9,
  no.~2, pp. 78--117, 1970.

\bibitem{smelyanskiy2016qhipster}
\BIBentryALTinterwordspacing
M.~Smelyanskiy, N.~P.~D. Sawaya, and A.~Aspuru-Guzik, ``{qHiPSTER}: the quantum
  high performance software testing environment,'' arXiv preprint
  arXiv:1601.07195, 2016. [Online]. Available:
  \url{https://arxiv.org/abs/1601.07195}
\BIBentrySTDinterwordspacing

\bibitem{haner2017half}
T.~H{\"a}ner and D.~S. Steiger, ``0.5 petabyte simulation of a 45-qubit quantum
  circuit,'' in \emph{Proceedings of the International Conference for High
  Performance Computing, Networking, Storage and Analysis (SC '17)}.\hskip 1em
  plus 0.5em minus 0.4em\relax ACM, 2017, pp. 1--10.

\bibitem{zhang2021hyquas}
C.~Zhang, Z.~Song, H.~Wang, K.~Rong, and J.~Zhai, ``{HyQuas}: hybrid
  partitioner based quantum circuit simulation system on {GPU},'' in
  \emph{Proceedings of the ACM International Conference on Supercomputing (ICS
  '21)}, 2021, pp. 443--454.

\bibitem{farhi2014qaoa}
\BIBentryALTinterwordspacing
E.~Farhi, J.~Goldstone, and S.~Gutmann, ``A quantum approximate optimization
  algorithm,'' arXiv preprint arXiv:1411.4028, 2014. [Online]. Available:
  \url{https://arxiv.org/abs/1411.4028}
\BIBentrySTDinterwordspacing

\bibitem{peruzzo2014variational}
A.~Peruzzo, J.~McClean, P.~Shadbolt, M.-H. Yung, X.-Q. Zhou, P.~J. Love,
  A.~Aspuru-Guzik, and J.~L. O'Brien, ``A variational eigenvalue solver on a
  photonic quantum processor,'' \emph{Nature Communications}, vol.~5, no.~1, p.
  4213, 2014.

\bibitem{li2023qasmbench}
A.~Li, S.~Stein, S.~Krishnamoorthy, and J.~Ang, ``{QASMBench}: A low-level
  quantum benchmark suite for {NISQ} evaluation and simulation,'' \emph{ACM
  Transactions on Quantum Computing}, vol.~4, no.~2, pp. 1--26, 2023.

\bibitem{villalonga2020establishing}
B.~Villalonga, D.~Lyakh, S.~Boixo, H.~Neven, T.~S. Humble, R.~Biswas, E.~G.
  Rieffel, A.~Ho, and S.~Mandr{\`a}, ``Establishing the quantum supremacy
  frontier with a 281 {Pflop/s} simulation,'' \emph{Quantum Science and
  Technology}, vol.~5, no.~3, p. 034003, 2020.

\bibitem{deraedt2007}
K.~De~Raedt, K.~Michielsen, H.~De~Raedt, B.~Trieu, G.~Arnold, M.~Richter,
  T.~Lippert, H.~Watanabe, and N.~Ito, ``Massively parallel quantum computer
  simulator,'' \emph{Computer Physics Communications}, vol. 176, no.~2, pp.
  121--136, 2007.

\bibitem{zhang2022uniq}
C.~Zhang, H.~Wang, Z.~Ma, L.~Xie, Z.~Song, and J.~Zhai, ``{UniQ}: a unified
  programming model for efficient quantum circuit simulation,'' in
  \emph{Proceedings of the International Conference for High Performance
  Computing, Networking, Storage and Analysis (SC '22)}, 2022, pp. 1--16.

\bibitem{tabuchi2024lazy}
\BIBentryALTinterwordspacing
Y.~Teranishi, S.~Hiraoka, W.~Mizukami, M.~Okita, and F.~Ino, ``Lazy qubit
  reordering for accelerating parallel state-vector-based quantum circuit
  simulation,'' \emph{ACM Transactions on Quantum Computing}, vol.~6, no.~4,
  pp. 1--33, 2025. [Online]. Available: \url{https://arxiv.org/abs/2410.04252}
\BIBentrySTDinterwordspacing

\bibitem{zhong2025q2chemistry}
\BIBentryALTinterwordspacing
G.~Zhong, Y.~Fan, and Z.~Li, ``Scalable parallel simulation of quantum circuits
  on {CPU} and {GPU} systems,'' arXiv preprint arXiv:2509.04955, 2025.
  [Online]. Available: \url{https://arxiv.org/abs/2509.04955}
\BIBentrySTDinterwordspacing

\bibitem{li2020dmsim}
A.~Li, O.~Subasi, X.~Yang, and S.~Krishnamoorthy, ``Density matrix quantum
  circuit simulation via the {BSP} machine on modern {GPU} clusters,'' in
  \emph{SC20: International Conference for High Performance Computing,
  Networking, Storage and Analysis}, 2020, pp. 1--15.

\bibitem{nwqsimsoftware}
\BIBentryALTinterwordspacing
{Pacific Northwest National Laboratory}, ``{NWQ-Sim}: Northwest quantum circuit
  simulation environment,'' Software repository, 2026, accessed 2026-07-15.
  [Online]. Available: \url{https://github.com/pnnl/NWQ-Sim}
\BIBentrySTDinterwordspacing

\bibitem{tabuchi2023mpiqulacs}
A.~Tabuchi, S.~Imamura, M.~Yamazaki, T.~Honda, A.~Kasagi, H.~Nakao,
  N.~Fukumoto, and K.~Nakashima, ``{mpiQulacs}: A scalable distributed quantum
  computer simulator for {ARM}-based clusters,'' in \emph{2023 IEEE
  International Conference on Quantum Computing and Engineering (QCE)}, 2023,
  pp. 959--969.

\bibitem{steiger2018projectq}
D.~S. Steiger, T.~H{\"a}ner, and M.~Troyer, ``{ProjectQ}: an open source
  software framework for quantum computing,'' \emph{Quantum}, vol.~2, p.~49,
  2018.

\bibitem{mindquantum_docs}
\BIBentryALTinterwordspacing
{MindSpore}, ``{MindSpore Quantum Documentation},'' n.d., accessed: Jul. 14,
  2026. [Online]. Available:
  \url{https://www.mindspore.cn/mindquantum/docs/en/stable/index.html}
\BIBentrySTDinterwordspacing

\bibitem{yao_docs}
X.-Z. Luo, J.-G. Liu, P.~Zhang, and L.~Wang, ``{Yao.jl}: Extensible, efficient
  framework for quantum algorithm design,'' \emph{Quantum}, vol.~4, p. 341,
  2020.

\bibitem{kokkos2014}
H.~C. Edwards, C.~R. Trott, and D.~Sunderland, ``{Kokkos}: enabling manycore
  performance portability through polymorphic memory access patterns,''
  \emph{Journal of Parallel and Distributed Computing}, vol.~74, no.~12, pp.
  3202--3216, 2014.

\bibitem{openmp51spec}
\BIBentryALTinterwordspacing
{OpenMP Architecture Review Board}, ``{OpenMP} application programming
  interface, version 5.1,'' 2020. [Online]. Available:
  \url{https://www.openmp.org/spec-html/5.1/openmp.html}
\BIBentrySTDinterwordspacing

\bibitem{khronos2020sycl}
\BIBentryALTinterwordspacing
{Khronos SYCL Working Group}, ``{SYCL} 2020 specification, revision 11,'' 2025.
  [Online]. Available:
  \url{https://registry.khronos.org/SYCL/specs/sycl-2020/html/sycl-2020.html}
\BIBentrySTDinterwordspacing

\bibitem{vanzee2015blis}
F.~G. Van~Zee and R.~A. van~de Geijn, ``{BLIS}: a framework for rapidly
  instantiating {BLAS} functionality,'' \emph{ACM Transactions on Mathematical
  Software}, vol.~41, no.~3, pp. 14:1--14:33, 2015.

\end{thebibliography}

\end{document}